\documentclass[12pt,letterpaper]{article}
\usepackage[letterpaper,margin=1in]
{geometry}
\usepackage{indentfirst}
\usepackage{authblk}
\usepackage{setspace}
\usepackage{fancyhdr}
\usepackage{amsthm}
\usepackage{lineno} 

\usepackage{color}
\definecolor{darkred}{RGB}{100,0,0}
\definecolor{darkgreen}{RGB}{0,150,0}
\definecolor{mediumgreen}{RGB}{0,200,0}
\definecolor{darkblue}{RGB}{0,0,150}

\usepackage{soul}

\usepackage[markup=underlined]{changes}
\definechangesauthor[name={Jiyue}, color=orange]{JQ}
\usepackage{hyperref}
\hypersetup{
colorlinks =true,
citecolor  = blue, 
linkcolor = blue
}
\usepackage{amsmath, bbm, dsfont, amsfonts, xcolor, bm, amssymb}
\usepackage{censor}
\usepackage{cleveref}
\usepackage{chngcntr}
\usepackage{natbib, booktabs, multirow, array,longtable,enumitem}
\usepackage{subcaption}
\usepackage{etoc}
\usepackage{comment}

\newcommand{\indep}{\perp \!\!\! \perp}

\newcommand{\Tearly}{T^u}  
\newcommand{\Xearly}{X^u} \newcommand{\Uiptw}{U_1^{\text{iptw}}}

\newcommand{\fullest}{U^f(t;\theta^a;T_1^a, T_2^a, Z)} 
\newcommand{\est}{U^a(t;\theta^a;S_c,\pi, F)} 
\newcommand{\fullestone}{U^f(t;\theta^1; T_1^1,T_2^1, Z)}
\newcommand{\Monefa}{M_1^f(t;\theta^a;T_1^a, T_2^a)}
\newcommand{\Mthreefa}{M_3^f(t;\theta^a;T_1^a, T_2^a)}
\newcommand{\causalMc}{M_c^a}
\newcommand{\deltaearlyone}{\delta_1^u}
\newcommand{\deltaearlytwo}{\delta_2^u}
\newcommand{\deltaearlyoneone}{\delta_1^{u,1}}
\newcommand{\deltaearlytwoone}{\delta_2^{u,1}}
\newcommand{\Hearly}{H^u}

\usepackage{makecell,graphicx}

\newtheorem{lemma}{Lemma}
\newtheorem{assumption}{Assumption}
  
\title{Doubly Robust Estimation under Covariate Dependent Censoring in Semi-Competing Risks Data}
\author[1]{Jiyue Qin}
\author[2]{Yuyao Wang}
\author[1,2,3,*]{Ronghui Xu}

\affil[1]{\small Biostatistics and Bioinformatics, Herbert Wertheim School of Public Health and Human Longevity Science, University of California, San Diego}
\affil[2]{\small Department of Mathematics, University of California, San Diego}
\affil[3]{\small Halicioglu Data Science Institute, University of California, San Diego}
\affil[*]{Corresponding author: \texttt{rxu@health.ucsd.edu}}

\date{}  
\begin{document}
\maketitle
\pagenumbering{gobble}
\textbf{Running head}: AIPCW in Semi-Competing Risks
\clearpage
\section*{Abstract}

Semi-competing risks occur when individuals may experience a non-terminal event and a terminal event, where the terminal event censors the non-terminal event but not vice versa. In the presence of covariate-dependent censoring,  
augmented inverse probability of censoring weighting (AIPCW) framework has not been developed for the special setting of semi-competing risks which, unlike competing risks, have asymmetric event time structure and complex estimands. Using semiparametric theory on coarsened data, we carefully develop an AIPCW framework for semi-competing risks. When treatment effects are of interest, we further integrate this approach with augmented inverse probability of treatment weighting (AIPTW), yielding a framework for estimating causal estimands under covariate-dependent censoring. The resulting estimators are shown to be doubly robust. Using the proposed framework, we estimate treatment specific risks of: i) non-terminal event, ii) terminal event without the non-terminal event, and iii) terminal event following the non-terminal event. We evaluate the finite sample performance through simulations, and apply the method to  data from the Honolulu Asia Aging Study to assess the causal effects of midlife heavy drinking on late life cognitive impairment and mortality.




\textbf{Keywords}:  AIPCW; cumulative incidence function; 
machine learning; monotone coarsening.
\clearpage
\pagenumbering{arabic}
\setcounter{page}{1}  

\section{Introduction}

Semi-competing risks  occur frequently in medical research when an individual may experience a non-terminal event (e.g.~disease) and a terminal event (e.g.~death).  The terminal event
censors the non-terminal event but not vice versa \citep{fine2001semi}. Such data are commonly described as a three-state illness-death model (Figure \ref{fig:illness_death}), a special case of a multi-state model, where individuals start in the ``healthy" state; then they may either transition to the ``diseased" state followed by the ``dead" state, or directly transition to the ``dead" state. The model is completely specified by the three transition rates between states, also called transition intensities or hazards  \citep{xu2010statistical}.

In many applications, the treatment effect is of primary interest. A commonly targeted estimand is the average treatment effect, also known as the average causal effect. It quantifies the population-level impact of an intervention, answering the policy-relevant question \citep{rubin1974estimating,hernan2020causal}: \textit{What would be the average potential outcome in the population if everyone were treated versus if no one were treated?}  As a motivating example for semi-competing risks, consider an aging cohort study 
where researchers are interested in the causal effect of midlife heavy drinking on cognitive impairment and mortality. In this context, it is often of interest to define the causal effect as the contrast between treatment groups in the following quantities: risk of cognitive impairment, risk of death without cognitive impairment, and risk of death following cognitive impairment \citep{meira2019estimation, zhang2024marginal}. 

Beyond the complexity brought on by the semi-competing risks structure, analyses of such data face two additional challenges. First, when marginal quantities  are of interest, covariate-dependent censoring must be properly accounted for. Second, in observational studies, confounding due to non-random treatment assignment can lead to non-comparable treatment groups and biased effect estimates if not appropriately addressed.
To tackle these challenges, we develop a general framework for constructing doubly robust (DR) estimators of a broad class of causal estimands. As specific instances, we use this framework to develop novel estimators for several commonly used estimands including the risk of the terminal event following the non-terminal event. The proposed framework is flexible and can be readily applied to other causal estimands.

To handle covariate-dependent censoring, for univariate survival data, 
inverse probability  of censoring weighting (IPCW)  \citep{robins1993information} 
is commonly used, 
and  \cite{Rotnitzky2005inverse} further developed the augmented IPCW (AIPCW) method, which enjoys the so-called doubly robust property.
While the AIPCW has been extensively studied in the univariate survival setting, research extending it to more complex data structures remains limited. Competing risks data represent a special case of semi-competing risks and, 
in this setting,  \cite{lok2018estimation} directly applied the above AIPCW approach to estimation of the cumulative incidence function (CIF). However, these methods do not directly extend to semi-competing risks, which involve an asymmetric dependence structure between the event times and also give rise to more complex estimands such as the risk of the terminal event following the non-terminal event. To date, AIPCW type methods have not been systematically studied for semi-competing risks under a general framework. 
As part of this work, we carefully develop a general AIPCW framework for the semi-competing risks setting (termed AIPCW\textsuperscript{sc}), by formulating censoring as monotone coarsening and leveraging the semiparametric theory for coarsened data \citep{tsiatis2006semiparametric}. 

To account for confounding, the augmented inverse probability of treatment weighting (AIPTW) method \citep[Chapter 1]{van2003unified} has been well established. 
AIPTW methods have been broadly applied to univariate survival settings, leading to DR estimators for various causal estimands, including treatment-specific survival probability and restricted mean survival times \citep{zhang2012double,bai2013doubly, bai2017optimal, sjolander2017doubly,dukes2019doubly, hou2023treatment, luo2025doubly}.
For competing risks, \cite{lin2022doubly} and \cite{van2025doubly} extended these methods to estimate cause-specific restricted mean survival time and CIF, respectively, both under the restrictive assumption of independent censoring. Collectively, these studies focus on specific estimands and do not readily extend to the more general semi-competing risks setting. To the best of our knowledge, no prior work has developed DR estimators for treatment effects in semi-competing risks data under covariate-dependent censoring.

To address the above gap, we build on our proposed AIPCW\textsuperscript{sc} framework and introduce the causal-AIPCW\textsuperscript{sc} framework, which integrates confounding adjustment via AIPTW. The resulting estimators involve two sets of nuisance models: one set contains models for the non-terminal event time and the terminal event time, and the other set contains models for censoring time and propensity score (PS). The approach is DR in the population sense: the true estimand is identified if either of the data generating nuisance models is known. 
Under regularity conditions then, the resulting estimator is 
(termed \textit{model DR}) consistent and asymptotically normal (CAN) if either set of the nuisance models is correctly specified and both sets of nuisance estimators are asymptotically linear (e.g., when parametric or semiparametric models are used). Besides model DR, recent literature has also established the \textit{rate DR} property of such augmented inverse probability weighting (AIPW) estimators \citep{rotnitzky2021characterization, hou2023treatment, rava2023doubly, luo2025doubly}. Rate DR ensures consistency and asymptotic normality when both sets of the models converge to the truth but at possibly slower than root-$n$ rates, as long as their product rate is faster than root-$n$, a property particularly attractive when using nonparametric machine learning methods for nuisance estimation. 

Our method provides a general framework to obtain DR estimators for treatment effects in semi-competing risk settings under covariate-dependent censoring, formulated using potential outcomes and applicable to various causal estimands. In this paper, we use this framework to estimate treatment-specific cumulative transition rates. From these, we further estimate treatment-specific risks of i) the non-terminal event, ii) the terminal event without the non-terminal event, and iii) the terminal event following the non-terminal event. The causal effect is then naturally defined and estimated as the contrast between treatment groups in these risk quantities. 

The novelty of our work is multifaceted. First, we develop a unified DR framework that enables causal inference in semi-competing risks data under covariate-dependent censoring. This framework is flexible and accommodates a broad class of estimands. Second, we use the framework to estimate cumulative transition rates and subsequently, the relevant risk quantities. Estimation of cumulative transition rates is particularly useful since other estimands can all be expressed as functions of them \citep{meira2019estimation}, positioning our work as a foundation for broader estimation goals.  Finally, most prior studies only considered parametric and semiparametric methods for nuisance estimation while we also consider nonparametric machine learning methods to showcase the \textit{rate DR} property.

The remainder of the paper is organized as follows. Sections \ref{sec:aipcw} and \ref{sec:causal-AIPCW} present our AIPCW\textsuperscript{sc} and causal-AIPCW\textsuperscript{sc} frameworks, including notations, assumptions and their DR properties. In Section \ref{sec:estimation} we use the general framework to construct estimators of aforementioned causal estimands.  Section \ref{sec:simulation} evaluates the finite-sample performance of the proposed method through simulations. In Section \ref{sec:application} we apply our method to data from the Honolulu Asia Aging Study to investigate the causal effect of midlife heavy drinking on late-life cognitive impairment and mortality. Section \ref{sec:discussion} concludes with some remarks. Proofs, additional technical details and simulation results are provided in the Supplementary Material.

\section{AIPCW\textsuperscript{sc} via monotone coarsening}\label{sec:aipcw}
\subsection{Notations and assumptions}
We first describe our AIPCW\textsuperscript{sc} framework, which can be applied independently to addressing covariate-dependent censoring in a broad setting and serves as a stepping stone to the causal-AIPCW\textsuperscript{sc} framework described subsequently.  

Let $T_1$ and $T_2$ denote time to the non-terminal event  and terminal event, respectively. If an individual experiences the terminal event before the non-terminal event occurs, then  $T_1 = \infty$ by convention (see e.g., 
\citealt{xu2010statistical, zhang2024marginal}). 
Note that the random vector  $(T_1, T_2)$ takes values on $\{(t_1,t_2): 0< t_1 \leqslant t_2<\infty \text{ or } t_1=\infty, t_2>0 \}$ (Figure \ref{fig:joint_density}). That is, there is no probability mass in the lower wedge $t_2<t_1<\infty$. The joint distribution of $T_1$ and $T_2$ in this case is completely characterized by the following 
three transition rates:
\begin{align}
\lambda_1(t)&= \lim_{\Delta \to 0^+} { P (T_1 \in [t, t +\Delta) | T_1 \geqslant t, T_2 \geqslant t )}/{\Delta},\\
\lambda_2(t)&= \lim_{\Delta \to 0^+} {  P (T_2 \in [t, t +\Delta) | T_1 \geqslant t, T_2 \geqslant t )}/{\Delta},\\
\lambda_{3} (t | t_1)&= \lim_{\Delta \to 0^+}  {P(T_2 \in [t, t+\Delta) | T_1 = t_1, T_2 \geqslant t) }/{\Delta}.
\end{align}
Note that $\lambda_1(t)$ and $\lambda_2(t)$ are cause-specific hazard functions in the usual competing risks setting, for time to the non-terminal event and time to the terminal event without non-terminal event, respectively. The third transition rate,  $\lambda_3(t|t_1)$,  is a hazard function of the terminal event conditional on time to the non-terminal event being $t_1$. 
We define the following three cumulative transition rates:
$\Lambda_j(t) = \int_{0}^t\lambda_j(u)du$ for $j =1,2$, and $ \Lambda_3(t| t_1) = \int_{0}^t\lambda_3(u| t_1)du$.

Let $T = \min(T_1, T_2)$. Let $C$ denote the censoring time and $Z$ the baseline covariates. Let $X_1= \min(T, C)$, $X_2 = \min(T_2, C)$, $\delta_1= \mathbbm{1} (X_1=T_1)$ and $\delta_2 = \mathbbm{1}(X_2 = T_2)$ where $\mathbbm{1}(\cdot)$  denotes the indicator function.
We observe $O=(X_1, X_2, \delta_1, \delta_2, Z)$. 
We adopt the following counting process notation for $C$: 
$N_{c}(t)=\mathbbm{1}( X_{2} \leqslant t, \delta_{2}=0)$, 
$Y_{c}(t)=\mathbbm{1}( X_{2} \geqslant t)$. 
Then $M_c(t ; S_c)=N_{c}(t)-\int_{0}^{t} Y_{c}(r) \lambda_{c}(r| Z) dr$
is a martingale  with respect to the 
filtration $\{\mathcal{F}_t\}_{t\geqslant 0}=\sigma(Z,\mathbbm{1}(T_2\leqslant \mu), \mathbbm{1}(C\leqslant \nu):\mu \leqslant t,\nu \leqslant t)$,  
where $S_c$ and $\lambda_c$ are the conditional survival and hazard functions of $C$ given $Z$, respectively.  
Let $S$  denote the conditional survival function of  $T$ given $Z$, and $F(t_1, t_2|Z)$ the conditional cumulative distribution function (CDF) of $(T_1, T_2)$ given $Z$.

We are interested in a parameter $\theta$ indexing the distribution of the censoring-free data  $(T_1, T_2, Z)$, that can be identified from a censoring-free data estimating function possibly dependent on the time $t$.  Specifically, we assume the following, where the superscript ``$o$'' denotes the true value of the parameter.

\renewcommand{\theassumption}{\Alph{assumption}}

\newlist{subassumption}{enumerate}{1}
\setlist[subassumption,1]{
  label=\textup{\arabic*.},
  ref=\theassumption\arabic*}

\begin{assumption}\label{ass:esteq}
$U^c (t;\theta)\equiv U^c (t;\theta, T_1, T_2,Z )$ is a censoring-free data  estimating function such that $\theta^o$ is the unique solution to  $E\{U^c(t;\theta)\} = 0$. 
\end{assumption}

For example, consider the estimation of $\theta = \Lambda_1(t)$. In this case, 
$U^c(t; \theta) = dM_1^c(t;\theta)$, where $ M_1^c(t;\theta) = \mathbbm{1}(T \leqslant t, T_1\leqslant T_2)-\int_0^t\mathbbm{1}(T\geqslant u) d\Lambda_1(u)$ is a martingale with respect to the natural history filtration 
\citep{andersen1991non}.

We additionally assume the following,  where $\indep$ denotes statistical independence.

\begin{assumption}
\label{ass:aipcw} 
We have
\begin{subassumption}
  \item \label{ass:aipcw_censoring} (conditional independent censoring)
$C\indep (T_1, T_2)\mid Z$.  
\item \label{ass:aipcw_positivity} (strict positivity)
There exists $0 < \epsilon < 1$ such that $S_c(\tau |Z ) > \epsilon$ almost surely (a.s.) and $S(\tau| Z) > \epsilon$ a.s., where $\tau$ is the maximum follow-up time.
\end{subassumption}
\end{assumption}
The above assumptions are commonly used (see e.g., \citealt{luo2025doubly}), where Assumption \ref{ass:aipcw_censoring} 
allows the censoring time to depend on covariates and Assumption \ref{ass:aipcw_positivity} ensures that the distributions of $C$ and $T$ can be identified up to the maximum follow-up time. 

\subsection{Censoring as monotone coarsening}

We develop the AIPCW\textsuperscript{sc} framework by casting 
the problem of censoring as monotone coarsening under the semi-competing risks setting. While this has been done for univariate survival outcomes \citep{tsiatis2006semiparametric, Rotnitzky2005inverse}, to our best knowledge it has not been considered for semi-competing risks. 
Under the coarsened data framework \citep{heitjan1991ignorability, gill1997coarsening},
instead of the 
full data  that we wish to observe, we observe only a \emph{coarsened} version of it, i.e.~a many-to-one function of it. 
Under the assumption of coarsening at random (CAR), to be defined below, semiparametric theory for coarsened data provides a systematic way for constructing CAN estimators and for improving efficiency \citep{tsiatis2006semiparametric}. As described therein, any observed data influence function is the sum of an IPW transformation of a full data influence function plus an augmentation term, and
one can obtain the most efficient observed data influence function within the class for a given full data influence function by 
projecting the IPW term onto the augmentation space.
The resulting influence function has improved efficiency over the IPW 
because this optimal augmentation term uses information available from coarsened individuals rather than relying solely on IPW among fully observed individuals.
Under the special case of {\it monotone coarsening}, where the coarsening probability can be modeled using the hazard of coarsening, this {optimal} augmentation term takes a closed form, without the need of successive approximation otherwise \citep[Chapter 10]{tsiatis2006semiparametric}. 
Such augmented estimators,  referred to as augmented inverse probability weighted complete-case (AIPWCC) estimators, often enjoy double robustness.  
In addition,
when the full data model is nonparametric,  the full data influence function is unique, and the influence function constructed in the above approach is the efficient influence function. 

In the semi-competing risks setting considered in this paper, censoring can be formulated as a coarsening mechanism. We show below that this mechanism is indeed monotone and that Assumption  \ref{ass:aipcw_censoring}  implies CAR. Consequently, the semiparametric theory for monotone coarsened data described above can be applied to deriving the AIPCW\textsuperscript{sc} estimating function, which improves efficiency and turns out to be doubly robust. We present the key steps in our development of  AIPCW\textsuperscript{sc} as follows, with details in  Supplementary Material.

First, to formulate our censoring problem as coarsening, 
let $\Tearly$ 
denote the earliest follow-up time at which the value of the estimating function
$U^c(t;\theta)$ becomes fully observed; that is, 
$U^c(t;\theta)$ depends on $(T_1, T_2)$ only through 
$(\Tearly, T_1\wedge \Tearly, \deltaearlyone, \deltaearlytwo)$, where $\deltaearlyone = \mathbbm{1}(T_1\leqslant \Tearly)$ and $\deltaearlytwo = \mathbbm{1}(T_2\leqslant \Tearly)$. Therefore,
$
U^c(t;\theta) = U^c (t;\theta, T_1, T_2,Z ) = U^c(\theta; W)$, where $W=(\Tearly, T_1\wedge \Tearly, \deltaearlyone, \deltaearlytwo, Z)$.


Continuing with the example  where  $\theta =\Lambda_1(t)$ and  
$U^c(t; \theta) =   \mathbbm{1}(T = t, T_1\leqslant T_2) - \mathbbm{1}(T\geqslant t) d\Lambda_1(t)$,
we have 
$\Tearly = \min(T,t)$ and  $U^c(\theta; W) = \mathbbm{1}(\Tearly= t)\{\deltaearlyone - d\Lambda_1(t)\}$. 
To see why $\Tearly$ takes such form, note that $\{C>\min(T,t)\} = \{C>T\}\cup \{t<C\leqslant T\}$. When $C>T$, we observe $T$ and  $\mathbbm{1}(T_1\leqslant T_2)$ and thus we can evaluate $U^c$; when $t<C\leqslant T$, we do not observe $T$ but know that  no event has occurred by time $t$ and thus $U^c = 0$. On the other hand, when $C<\min(T,t)$, we  do not know whether an event would have occurred by time $t$ and thus cannot evaluate $U^c$.

In general, $\Tearly$ is a function of $(T_1,T_2,t)$ and its form depends on the specific form of $U^c(t;\theta)$. 
While in the example above 
$\Tearly = \min(T, t)$ 
is the same as in the usual competing risks setting, this is not always the case and 
$\Tearly$ may exceed $\min(T, t)$. One such example is $\theta = \Lambda_3(t)$ 
in Section \ref{sec:estimation}, for which $\Tearly = \min(T_2, t)$. 
Since any $U^c$ of interest is always observed when $C>T_2$, we always have $\Tearly\leqslant T_2$. Under the coarsened data framework \citep[Chapter 9]{tsiatis2006semiparametric},  
$W$ is the data that we wish to observe and $(R, G_R(W))$ denotes the coarsened data  we actually observe. Here, $R$ is the coarsening variable such that when $R=r$ ($0<r<\infty$), we observe a many-to-one mapping $G_r(W)$, instead of $W$ itself. By convention, $R=\infty$ denotes no coarsening. 
Specifically, we define $(R, G_R(W))$ as: 
\begin{align}\label{eq:monocoar}
& \left\{\begin{array}{l}
\text { if } C  >  \Tearly , \text { then } R=\infty, G_{\infty}(W)=W; \\
\text { if } C=r \leqslant \Tearly, \text { then } R=r, G_{r}(W)= \{\mathbbm{1}(\Tearly\geqslant  r),Z\}.
\end{array} \right. 
\end{align}
Indeed, if $C> \Tearly$, then $\Tearly$ is observed and the value of $U^c$ can be determined;  and the second line in \eqref{eq:monocoar} reflects the fact that when censoring occurs at time $r$ before $\Tearly$ is reached, the value of $U^c(t;\theta)$ is not observed, and the only information we have is that $\Tearly \geqslant r$. This formulation differs from that in \citet{tsiatis2006semiparametric} in that our formulation applies to the semi-competing risks setting, 
and 
because $\Tearly$ depends on $t$, the resulting  weights are also time dependent. 
In contrast, the formulation in \citet{tsiatis2006semiparametric} uses a coarsening variable that is not indexed by $t$, which leads to  weights that do not vary with time. Note that time dependent weights are  also used in the AIPCW method of \cite{Rotnitzky2005inverse} for univariate survival data and can make more efficient use of the data in general.

Next, to apply the semiparametric theory on monotone coarsened data under CAR, we first provide the precise definitions of these two concepts. 
Monotone coarsening assumes that for any $r<s$, there exists a fixed function $f_{r,s}$  such that $G_r(W) = f_{r,s}(G_s(W))$, while  
CAR assumes that the distribution of the coarsening variable depends on censoring-free data only as a function of the observed data: 
 $P(R=r| W) = w(r,G_r(W))$.  
Here $P(R=r|W)$ denotes the conditional probability mass when $r=\infty$ and the conditional probability density when $r<\infty$. 
In the Supplementary Material we show that 1)
when formulating semi-competing risks data subject to censoring  as coarsened data in \eqref{eq:monocoar}, such coarsening mechanism is monotone, 
and  2) Assumption \ref{ass:aipcw_censoring} 
implies CAR.  
Following \citet[Chapter~9]{tsiatis2006semiparametric} then, the AIPW estimating function for monotone coarsened data under CAR has an IPW term plus an augmentation term:
\begin{equation}\label{eq:coarsen}
\frac{\mathbbm{1}(R=\infty) U^{c}(\theta;W)}{P(R=\infty|W)}+ \int_{0}^\infty \frac{L_r(G_{r}(W))}{K_{r}(W)}dM_R(r),
\end{equation} 
where $M_R(r) = \mathbbm{1}(R\leqslant r)- \int_0^r \mathbbm{1}(R \geqslant u)\lambda_u(W)du$ is a martingale \cite[Theorem 9.2]{tsiatis2006semiparametric}, 
$\lambda_r(W) = \lim_{\Delta \to 0^+}  {P(R \in [r, r+\Delta) | R\geqslant r, W) }/{\Delta}$, $K_r(W)= P(R>r | W)$ and  $L_r(G_{r}(W))$ is an arbitrary function of $G_{r}(W)$, 
yielding a class of observed data estimating functions corresponding to this particular censoring-free data estimating function $U^c$.
Note that the original presentation in \citet[Chapter~9]{tsiatis2006semiparametric} is given for discrete $R$ where the augmentation term involves a sum over $r$, while  \eqref{eq:coarsen} is the natural continuous time analog of that original formulation.

We had earlier that  $ \mathbbm{1}(R=\infty) = \mathbbm{1}(\Tearly< C)$, and $ U^c(\theta; W) = U^c(t;\theta) $. 
For the rest of the quantities in \eqref{eq:coarsen} we show in the Supplementary Material that 
 $P(R=\infty | W) = S_c(\Tearly | Z)$,  
$K_r(W)= S_c(\Tearly \wedge r| Z)$  where $ \Tearly \wedge r = \min(\Tearly, r) $, and $d M_R(r) = \mathbbm{1}( \Xearly\geqslant r) dM_{c}(r)$ where $\Xearly=\min(\Tearly, C)$.
By analogy with Theorem 10.4 in \cite{tsiatis2006semiparametric} 
but for continuous time, we obtain within the class of observed data estimating functions the most efficient  estimator  by taking $L_r(G_{r}(W))=E\{U^{c}(t ; \theta) | \Tearly \geqslant r, Z\} \equiv L(r, t ; \theta ; F ; Z)$.
Formula \eqref{eq:coarsen} then gives
 the following AIPCW\textsuperscript{sc} estimating function:  
\begin{equation}\label{eq:aipcw}
U(t;\theta;S_c, F)= 
\frac{\mathbbm{1}(\Tearly< C) U^c(t ; \theta)}{S_c(\Xearly | Z)}
+\int_0^{\Xearly} \frac{L(r, t ; \theta ; F;Z)}{S_c(r | Z)} d M_c(r ; S_c),
\end{equation}
where $M_c(t ; S_c)=N_{c}(t)-\int_{0}^{t} Y_{c}(r) \lambda_{c}(r| Z) dr$ as defined before. 
Again it consists of two components: an IPCW term, which weights each uncensored individual by the inverse probability of being uncensored, 
and an augmentation term which represents the 
contribution to the estimation from a censored subject.
Note that \eqref{eq:aipcw} involves two nuisance parameters $S_c$ and $F$. 
Also note that setting $T_1=\infty$ in \eqref{eq:aipcw} returns the univariate AIPCW in \cite{Rotnitzky2005inverse},  where the single survival outcome is the terminal event.

The following lemma shows that $\theta$ can be identified and $U(t;\theta;S_c, F)$ is a DR estimating function (proof in 
Supplementary Material). 

\begin{lemma}[Identifiability and double robustness of AIPCW\textsuperscript{sc}]\label{lemma:DR}
Under Assumption \ref{ass:aipcw},\\
$E\{U(t;\theta;S_c, F)\} = E\{U^c(t;\theta)\}$ for all $\theta$, 
if either $F= F^o$ or $S_c=S_c^o$; hence 
$\theta^o$ is the unique solution to $E\{U(t;\theta;S_c, F)\} = 0$ under
Assumption \ref{ass:esteq}. 
\end{lemma}

\section{Causal-AIPCW\textsuperscript{sc}}\label{sec:causal-AIPCW}
We now consider estimating treatment effects. Following the potential outcomes framework, let  $A =0, 1$ be the binary treatment assignment, possibly not randomized; $Z$ be a vector of covariates; $T_1^a$, $T_2^a$ and $C^a$ be the potential non-terminal event time, terminal event time, and censoring time had the subject received treatment $a$, where $a=0,1$. Let $\lambda_1^a(t)$, $\lambda_2^a(t)$, $\lambda_3^a(t|t_1)$ denote the transition rates corresponding to the counterfactual states for $a=0,1$. Let $T^a = \min (T_1^a, T_2^a)$. Let $T$, $T_1$, $T_2$ and $C$ be their corresponding ``factual" counterparts once the treatment is received. The variables $X_1$, $X_2$, $\delta_1$, $\delta_2$ are defined the same as in the previous section. With a slight abuse of notation, we now observe $O=(X_1,X_2,\delta_1, \delta_2,A,Z)$. Let $S_c, \lambda_c$,  $S$ and $F$ denote the conditional survival function of $C$ given $(A, Z)$, the conditional hazard function of $C$ given $(A, Z)$,  the conditional survival function of  $T$ given $(A, Z)$ and the conditional CDF of $(T_1, T_2)$ given $(A, Z)$, respectively. Let $\pi = P(A=1|Z)$ denote the propensity score. 

In the hypothetical world without  censoring and with fully observed potential outcomes, we have the full data $(T_1^a,T_2^a,Z)$, $a=0, 1$. The goal is to estimate $\theta^a$, an estimand related to the distribution of the potential outcomes $ (T_1^a,T_2^a) $, that can be identified from a full data estimating function  $\fullest$. Specifically, we assume the following.
\begin{assumption}\label{ass:causal_esteq}
    $\theta^{a,o}$ is the unique solution to $E\{ \fullest\} = 0$.
\end{assumption}

For example, if we are interested in estimating 
$\theta^a = \Lambda_1^a(t)$, then $\fullest = d\Monefa$ where $\Monefa = \mathbbm{1}(T^a \leqslant t, T_1^a\leqslant T_2^a)-\int_0^t\mathbbm{1}(T^a\geqslant u) d\Lambda_1^a(u)$. 

We use the following additional assumptions, which are standard in causal inference 
under covariate-dependent censoring (see e.g., \citealt{hernan2020causal, luo2025doubly}): 
\begin{assumption}
\label{ass:c-aipcw} 
We have
\begin{subassumption}
  \item \label{ass:causal_sutva}(stable unit treatment value assumption, SUTVA)
There is only one version of treatment and there is no interference among the subjects.
\item \label{ass:causal_consistency}(consistency)
If $A=a$, then $T_1=T_1^{a}, T_2=T_2^{a}, C=C^{a}$, for $a=0,1$.
\item \label{ass:causal_exchange} (exchangeability)
 $(T_1^a, T_2^{a}, C^{a})\indep A \mid Z$, for $a=0,1$.  
 \item \label{ass:causal_positive}(strict positivity)
There exists $0 < \epsilon < 1$ such that $\epsilon < \pi(Z)  < 1 - \epsilon$ a.s., $S_c(\tau |A,Z ) > \epsilon$ a.s. and $S(\tau| A , Z) > \epsilon$ a.s., where $\tau$ is the maximum follow-up time. 
\item 
\label{ass:causal_censoring} (conditional independent censoring)
 $(T_1^a, T_2^{a}) \indep C^{a} \mid  Z$, for $a=0,1.$   
\end{subassumption}
\end{assumption}
     
We now develop an approach that maps a  full data estimating function to an observed data estimating function, referred to as the causal-AIPCW\textsuperscript{sc} method.
To derive the causal-AIPCW\textsuperscript{sc} estimating function, we first apply the AIPTW method as described in \cite{van2003unified} to $\fullest$, to obtain estimating functions $U_a^\text{aiptw}(t;\theta^a;\pi, F)$ for the censoring-free data $(T_1,T_2,A,Z)$. Details of this step are provided in Supplementary Material. Under our assumptions, we have $C\indep (T_1,T_2)\mid (A,Z)$. Thus, we can apply our AIPCW\textsuperscript{sc} method from the previous section to $U_a^\text{aiptw}(t;\theta^a;\pi, F)$, in order to account for censoring. The resulting causal-AIPCW\textsuperscript{sc} estimating function for $\theta^a$ is:
\begin{align}\label{eq:c-aipcw}
U^a(t;\theta^a;S_c,\pi, F)&= 
\frac{A^a(1-A)^{1-a}}{\pi(Z)^a\{1-\pi(Z)\}^{1-a}}\frac{ \mathbbm{1}(\Tearly < C)} {S_c(\Xearly | A=a, Z)} U^f(t ; \theta^a) \nonumber\\
&\quad 
+\left[1-\frac{A^a(1-A)^{1-a}}{\pi(Z)^a\{1-\pi(Z)\}^{1-a}} \right] E\{U^f(t ; \theta^a) | A=a, Z\}\nonumber\\
+ & \frac{A^a(1-A)^{1-a}}{\pi(Z)^a\{1-\pi(Z)\}^{1-a}} \int_0^{\Xearly} \frac{E\{U^f(t;\theta^a) | \Tearly \geqslant r, A=a, Z\}}{S_c(r | A=a, Z)} d\causalMc(r; S_c),
\end{align}
where $U^f(t;\theta^a) \equiv U^f(t;\theta^a; T_1,T_2, Z)$, 
$\Tearly$ denotes the earliest time point such that the estimating function $U^f(t;\theta^a)$ can be observed,  and 
$\causalMc(t ; S_c)=\mathbbm{1}(X_2\leqslant t, \delta_2 = 0)-\int_0^t \mathbbm{1}(X_2 \geqslant u) \lambda_c(u | A,Z) du$. 
Formula \eqref{eq:c-aipcw} consists of three components: the first is the inverse probability weighted (IPW) term, which weights each treated and uncensored individual by the inverse probability of receiving their treatment and remaining uncensored; the other two are augmentation terms, capturing the contributions from i) a participant who did not receive the treatment and ii) a participant who received the treatment but was censored, respectively. It involves three nuisance parameters: $S_c$, $\pi$ and $F$. 
Note that \eqref{eq:c-aipcw} can also be derived by formulating the problem as monotone coarsening, with treatment introducing an additional layer of coarsening. See Supplementary Material 
for the exact formulation. From this perspective, the causal-AIPCW\textsuperscript{sc} estimator also belongs to the broader class of AIPWCC estimator.

The following lemma establishes that $\theta^a$ can be identified and $\est$ is a DR estimating function (proof in 
Supplementary Material).

\begin{lemma}[Identifiability and double robustness of causal-AIPCW\textsuperscript{sc}]\label{lemma:causalDR}
Under Assumption \ref{ass:c-aipcw}, 
for $a=0,1$, $E\{U^a(t;\theta^a;S_c,\pi, F)\} =E\{\fullest\}$ for all $\theta^a$ if either 1) $F= F^o$ or 2) $S_c=S_c^o$ and $  \pi =\pi^o$; hence 
$\theta^{a,o}$ is the unique solution to $E\{\est\} = 0$ under Assumption \ref{ass:causal_esteq}. 
\end{lemma}

\section{Estimation of relevant quantities} \label{sec:estimation}

The above estimating functions involve nuisance parameters, which are usually unknown and need to be estimated. 
When parametric or semiparametric methods are used for nuisance parameter estimation, 
with a random sample $\{O_i\}_{i=1}^n$ the estimator $\hat{\theta}^a$ is obtained by solving
$\sum_{i=1}^n U_i^a(t;\theta^a ;\hat{S}_c, \hat{\pi},\hat{F})=0$, where $ \hat{S}_c, \hat{\pi}$ and $\hat{F}$ are obtained using the same sample.
Nonparametric methods are more flexible but typically converge at slower than root-$n$ rate and are not asymptotically linear. To ensure root-$n$ consistency of $\hat \theta^a$ in such cases, we use the $K$-fold cross-fitting procedure \citep{ robins2008higher, chernozhukov2018double} 
which introduces independence between the nuisance estimates and the data used to evaluate the estimating function, thereby avoiding the requirement of asymptotic linearity. 
Specifically, the data are split into $K$ folds of similar size with index sets $\mathcal{I}_1, \ldots, \mathcal{I}_K$. For the $k$-th fold $(k=1, \ldots, K)$, we use data out of the $k$-th fold to estimate the nuisance parameters, denoted by 
$\hat{S}_c^{(-k)}$, $ \hat{\pi}^{(-k)}$ and $ \hat{F}^{(-k)}$. The estimator $\hat{\theta}^a$ is then obtained by solving
\begin{equation}
\sum_{k=1}^K \sum_{i \in \mathcal{I}_k} U_i^a\left(t;\theta^a ;  \hat{S}_c^{(-k)}, \hat{\pi}^{(-k)}, \hat{F}^{(-k)} \right)=0.
\end{equation}

Under our causal-AIPCW\textsuperscript{sc} framework, we develop estimators for several 
key causal estimands that are easy to interpret, namely, the following three cumulative incidence functions (CIFs) that represent the treatment-specific risks of i) the non-terminal event, ii) the terminal event without the non-terminal event, and iii) the terminal event following the non-terminal event: 
$\mathrm{CIF}_1^a(t) = P(T_1^a\leqslant t, T_1^a\leqslant T_2^a)$, 
$\mathrm{CIF}_2^a(t) = P(T_2^a\leqslant t, T_1^a>T_2^a)$,  and
$\mathrm{CIF}_3^a(t|t_1) = P(T_2^a\leqslant t| T_1^a=t_1)$.
Note that $\mathrm{CIF}_1$ and $\mathrm{CIF}_2$ are standard quantities (also known as subdistribution functions) in the usual competing-risks setting, whereas $\mathrm{CIF}_3$ is a conditional CDF of the terminal event given that the non-terminal event occurred at $t_1$. To estimate $\mathrm{CIF}_3$, we make 
the commonly used Markov assumption (see e.g., \citealt{xu2010statistical,meira2019estimation, zhang2024marginal}): $\lambda_{3}^a (t |t_1) = \lambda_{3}^a (t)\mathbbm{1}(t\geqslant t_1)$.  
Under this condition, $\mathrm{CIF}_3^a(t|t_1) = P(T_2^a\leqslant t| T_1^a\leqslant t_1, T_2^a\geqslant t_1)$, which is the estimand investigated in \cite{zhang2024marginal}.
Then treatment effects can be naturally defined as the contrasts between $\mathrm{CIF}_j^0$ and $\mathrm{CIF}_j^1, j=1,2,3$, such as risk differences. 

The above CIFs can be expressed in terms of $\Lambda_j^a, j=1,2,3$ as follows: 
$\mathrm{CIF}_j^a(t) = \int_0^tS^a(u)d\Lambda_j^a(u)$, $ j = 1,2$; and $ \mathrm{CIF}_3^a(t| t_1) = 1-\exp\{ \Lambda_3^a(t_1)-\Lambda_3^a(t)\}$, 
where $S^a(t)\equiv P(T^a > t) = \exp \{-\Lambda_1^a(t)-\Lambda_2^a(t) \}.$ 
These expressions for $\mathrm{CIF}_1$ and $\mathrm{CIF}_2$ are well-known in competing risks (see e.g. \citealt{klein2006survival} Chapter 2) and the derivation of $\mathrm{CIF}_3$ is provided in the Appendix. Now we consider estimation of  $\Lambda_j^a, j=1,2,3$ using the proposed causal-AIPCW\textsuperscript{sc} method.

To estimate $\Lambda^a_1(t)$, as indicated before we have the following: $\theta^a= \Lambda^a_1(t)$, $\fullest = d\Monefa$,
and $\Tearly = \min(T, t).$ Plugging into the causal-AIPCW\textsuperscript{sc} formula \eqref{eq:c-aipcw}, we obtain the following closed-form solution of $\hat{\Lambda}_1^1(t)$:
\begin{equation}\label{eq:Lambda1}
\hat{\Lambda}_1^1(t) = \int_0^t\frac{
\sum_{i=1}^n Q_{1i}^1(u)dN_{1i}(u)+ Q_{2i}^1(u)\hat{S}(u | A=1, Z_i) d\hat\Lambda_1(u | A=1, Z_i) }{
\sum_{i=1}^n Q_{1i}^1(u)Y_{1i}(u)+ Q_{2i}^1(u) \hat{S}(u | A=1, Z_i)},
\end{equation} 
where $N_{1i}(t) = \mathbbm{1}(X_{1i}\leqslant t, \delta_{1i}=1)$, $Y_{1i}(t) = \mathbbm{1}( X_{1i} \geqslant t)$, $\tilde{X}_1= \min(X_1,t)$,
$
Q_{1i}^1(t) = {A_i} / \{\hat{\pi}(Z_i) \hat{S}_c(t|Z_i,A=1)\}$, 
$
Q_{2i}^1(t) = 1 
+ {A_i} \{ m_1(\delta_{2i}, X_{1i}, X_{2i}, 1, Z_i, t) -1 \} / {\hat{\pi}(Z_i)}$, and 
\begin{equation}
m_1(\delta_2, X_1, X_2, A, Z, t) = \frac{(1-\delta_2) \mathbbm{1}(X_2\leqslant \tilde{X}_1) }{ \hat{S}(X_2 | A, Z) \hat{S}_c(X_2| A, Z) } +\int_0^{\tilde{X}_1} \frac{d \hat{S}_c(r | A, Z) }{\hat{S}(r | A, Z) \hat{S}_c^{2}(r| A, Z) }.
\end{equation}

The estimator for $\Lambda_1^0(t)$ can be derived similarly, and by symmetry, estimation of $\Lambda^a_2(t)$ follows analogously. The expressions of these estimators are provided in the Appendix, with derivation details included in the 
Supplementary Material.

To estimate $\Lambda_{3}^a(t)$, we have the following: $\theta^a= \Lambda^a_3(t)$, 
$\fullest = d\Mthreefa$ where $\Mthreefa = \mathbbm{1}(T_1^a\leqslant T_2^a\leqslant t)-\int_0^t\mathbbm{1}(T_1^a< u \leqslant T_2^a) d\Lambda_3^a(u)$, 
 $\Tearly = \min(T_2, t).$ We have 
\begin{equation}\label{eq:Lambda3}
\hat{\Lambda}_3^1(t)=\int_{0}^{t}\frac{
\sum_{i=1}^n Q_{1i}^1(u)dN_{3i}(u)+ Q_{3i}^1(u)  \{-d\hat S_2(u|A=1, Z_i)-\hat S(u|A=1,Z_i)d\hat \Lambda_2(u|A=1, Z_i)\} }{
\sum_{i=1}^n Q_{1i}^1(u)Y_{3i}(u)+ Q_{3i}^1(u)\{\hat S_2(u|A=1,Z_i)-\hat S(u|A=1,Z_i)\}},
\end{equation}
where $N_{3i}(t) = \mathbbm{1}(X_{2i}\leqslant t, \delta_{1i}=1, \delta_{2i}=1)$, $Y_{3i}(t) = \mathbbm{1}(X_{1i}< t\leqslant X_{2i}, \delta_{1i}=1)$, $S_2 (t|A, Z) = P(T_2>t|A, Z)$, 
$Q_{1i}^1$ is defined the same as in \eqref{eq:Lambda1}, $\tilde{X}_2=\min(X_2, t)$,\\
$ 
Q_{3i}^1(t) = 1 + {A_i} \{ m_3(\delta_{2i}, X_{2i}, 1, Z_i, t) -1 \} / {\hat{\pi}(Z_i)} $, and 
\begin{equation}
m_3(\delta_2, X_2, A, Z, t)
=\frac{(1-\delta_2)\mathbbm{1}(X_2 \leqslant \tilde{X}_2)}
{\hat{S}_2(X_2 | A, Z)\, \hat{S}_c(X_2 | A, Z)}
+ \int_0^{\tilde{X}_2}
\frac{d\hat{S}_c(r | A, Z)}
{\hat{S}_2(r | A, Z)\, \hat{S}_c^2(r | A, Z)}.
\end{equation}
The derivation details are provided in the Supplementary Material.
The solution for $\hat{\Lambda}_3^0(t)$ can be derived similarly and is provided in the Appendix. 

Under additional regularity conditions, 
the double robustness of the estimating functions implies that the corresponding estimators are CAN when either of the two sets of nuisance models is correctly specified parametrically or semiparametrically. When flexible nonparametric estimators are used, the method inherits rate double robustness properties established in recent literature (see e.g., \citealt{luo2025doubly}). 
In addition, 
since the three transition rates $\Lambda_1^a$, $\Lambda_2^a$, and $\Lambda_3^a$ are variationally independent, the observed data likelihood can be factored as a product of likelihoods involving the three transition rates, respectively. Therefore the observed data tangent space can be decomposed as a direct sum of three mutually orthogonal spaces $\mathcal{T}_1 \oplus \mathcal{T}_2 \oplus \mathcal{T}_3$, each corresponding to one transition rate 
\citep{tsiatis2006semiparametric}. 
Because of the above, the  estimation for $\Lambda_1^a$ and $\Lambda_2^a$ under the nonparametric model  remains efficient under the Markov model for $\Lambda_3^a$; that is, the proposed estimators $\hat\Lambda_1^a$ and $\hat\Lambda_2^a$
are locally semiparametric efficient when both sets of nuisance models are correctly specified.

For comparison, we provide in the Supplementary Material expressions of  the corresponding IPW estimators that are not doubly robust, as well as the naive Nelson-Aalen estimators which ignore confounding and covariate-dependent censoring.

\section{Simulations}\label{sec:simulation}

We conducted simulations to evaluate the finite sample performance of our proposed method. The semi-competing risks data were simulated using an approach adapted from \cite{zhang2024marginal} with details in the Supplementary Material. We considered six estimands:  $\operatorname{CIF}_1^a(3)$, $\operatorname{CIF}_2^a(3)$ and  $\operatorname{CIF}_3^a(3 | 0.5)$ for $a=0,1$. The true values of these estimands were computed from a simulated full data sample of size $10^7$. We compared our proposed causal-AIPCW\textsuperscript{sc} estimators with the IPW estimators, naive Nelson-Aalen estimators which do not account for covariate-dependent censoring and confounding,  and full data estimators. Note that the full data estimators, which use empirical probability estimates computed from the full data, are unattainable in practice and were included as a  benchmark. 

We considered various working models for the nuisance parameters, including intentionally misspecified models to evaluate method  robustness. 
For the censoring model, $S_c(t|A,Z)$ was estimated using Cox regression from  R package \texttt{survival} and random survival forest (RSF) \citep{Ishwaran2008random} from  \texttt{randomForestSRC}. For the PS model, $\pi(Z)$ was estimated using logistic regression and the gradient boosted model (GBM) from  \texttt{twang} \citep{ridgeway2022toolkit}. For the event time model, the nuisance parameters  are $\Lambda_1(t|A,Z), \Lambda_2(t|A,Z)$ and $S_2(t|{A,Z})$. The first two were estimated using cause-specific Cox regression from \texttt{survival} and RSF for competing risks \citep{Ishwaran2014random} from \texttt{randomForestSRC}, while $S_2(t|{A,Z})$ 
was estimated using Cox and RSF for univariate survival data in the same manner as for $S_c(t|A,Z)$. Five-fold cross-fitting was implemented when nonparametric methods were used for nuisance estimation. Standard errors (SEs) of the estimators were obtained from nonparametric bootstrap with 200 samples and 95\% confidence intervals (CIs) were constructed using the normal approximation. The number of simulation replicates was 1000, corresponding to a margin of error of  ±1.35\% for the coverage probability.

Table \ref{table:sim} presents simulation results from samples of size 500 for $\operatorname{CIF}_j^0, j =1,2,3$, including the bias, empirical standard deviation (SD), the mean of the bootstrapped SE, and the coverage probability (CP) of 95\% CIs. The results for $\operatorname{CIF}_j^1, j =1,2,3$ are similar and are provided in the Supplementary Material. 
Consistent with the \textit{model DR} property, the causal-AIPCW\textsuperscript{sc} estimators  performed well when one set of nuisance models was correctly estimated at the root-$n$ rate, yielding small bias and coverage close to the nominal 95\%, even when the other set was misspecified. The causal-AIPCW\textsuperscript{sc} estimators also showed good performance when both sets of nuisance models were estimated using machine learning techniques, thereby demonstrating the \textit{rate DR} property. 
As expected, the naive method and the IPW method with misspecified nuisance models performed poorly, exhibiting large bias and poor coverage.

\section{Application to HAAS Study}\label{sec:application}

We applied our method to data from the Honolulu-Asia Aging Study (HAAS), a prospective community-based cohort study of 3,734 Japanese-American men born 1900-1919 and residing on Oahu, Hawaii \citep{p2012honolulu}. HAAS was initiated in 1991 as a continuation of the Honolulu Heart Program (HHP, 1965-1973) and concluded in 2012. Our goal was to assess the effect of midlife heavy drinking on cognitive impairment as well as death, which are semi-competing risks. Alcohol consumption was assessed based on self reported daily intake of beer, wine, liquor, and sake, and quantified as ethanol intake. It was measured during the HHP study, when the participants were in midlife, and was dichotomized into two groups: heavy  drinking (consuming $>1.2$ ounces of ethanol per day at any point during midlife) versus not. Cognitive impairment was assessed using scores from the Cognitive Assessment and Screening Instrument (CASI) and a score $<74$ was considered a moderate impairment (MI). Following \cite{zhang2024marginal}, we considered the following five baseline covariates measured at the start of HAAS: age, years of education, ApoE genotype
(positive/negative), systolic blood pressure, and heart rate. We restricted the analysis to participants with normal cognitive function at baseline (i.e., CASI $\geqslant$ 74), and after excluding those with missing exposure or covariates, the final sample included 1,881 participants with follow-up since the start of HAAS. Among them, 431 participants (22.9\%) were heavy drinkers
and 1,450 (77.1\%) were non-heavy drinkers. The number of MI and death events in the overall sample and stratified by the exposure group are listed in Table \ref{table:haas_event}. 
 
We applied the proposed causal-AIPCW\textsuperscript{sc} estimators to analyze the causal effect of midlife heavy drinking on the risks of MI, death without MI, and death following MI using HAAS data, with the time origin being the study entry of HAAS. 
Specifically, we estimated $\operatorname{CIF}_1^a(t)$,  $\operatorname{CIF}_2^a(t)$,  $\operatorname{CIF}_3^a(t | 8)$ for $a=0,1$ and $t=5,10,15,20$ years. These correspond to the risk of MI before $t$, the risk of death without MI before $t$, and the risk of death before $t$, given that MI occurred at 8 years, respectively, for heavy-drinkers and non-heavy drinkers. For the nuisance estimation, we used GBM for the PS model, RSF for both the censoring model and event time model. Bootstrap was used to construct 95\% CIs as described in the simulations. Risk differences (RDs) and their CIs were also computed. 

As shown in Figure \ref{fig:haas} and Supplementary Table S2
, under the previously described assumptions including no unmeasured confounding, heavy drinking caused a higher risk of MI 
at 10 and 15 years (RD = 0.07, 95\% CI: 0.02-0.12 at 10 years; RD = 0.07, 95\% CI: 0.01-0.12 at 15 years), but not at 5 or 20 years. Heavy drinking did not have a significant effect on the risk of death without MI or death following MI by 8 years at the evaluated time points.

Our findings align with 
 \cite{zhang2024marginal}  
in showing  significant effects of midlife heavy drinking on MI risk at 10 and 15 years after study entry, but differ at 5 years. This discrepancy likely stems from the restrictive assumptions they impose with the Cox marginal structural model as well as independent censoring. 

\section{Discussion}\label{sec:discussion}

We have developed a general framework for doubly robust causal estimation under covariate-dependent censoring in semi-competing risks settings, applicable to a broad range of causal estimands.  Using this framework, we have proposed estimators for treatment-specific risks of i) the non-terminal event, ii) the terminal event without non-terminal event, and iii) the terminal event following the non-terminal event. These quantities naturally lead to the estimated treatment effects, defined as contrasts between the treatment-specific risks. The proposed estimators enjoy double robustness, and we have demonstrated their  performance in simulations. Our application to data from the HAAS study showcases the method's practical utility. The method has been implemented in the R package 
\texttt{causalDR} available on GitHub. 

We have adopted the commonly used Markov assumption $\lambda_{3}^a (t |t_1) = \lambda_{3}^a (t)\mathbbm{1}(t\geqslant t_1)$ to define $\Lambda_{3}^a$ and $\operatorname{CIF}_3^a$. Importantly, this assumption applies only to these  quantities and is not required for our general framework or the other estimands we considered. Recent developments in multi-state modeling literature have introduced methods for estimating state transition probabilities under non-Markov settings \citep{putter2018non, andersen2022inference}. 
Future work could explore estimating these quantities using our framework under non-Markov settings. Nonetheless, as reported in \cite{zhang2024marginal}, the Markov assumption was not rejected in the HAAS data, supporting our estimators' practical relevance.

It is also worth noting that the Markov condition we impose is at the marginal level (i.e., without conditioning on covariates). In general, a conditional Markov assumption does not imply the marginal Markov condition (proof in Supplementary Material). Therefore we do not impose a conditional Markov assumption of the form $\lambda_{3}^a (t | t_1, A, Z) = \lambda_{3}^a (t| A,Z)\mathbbm{1}(t\geqslant t_1)$, even though doing so would simplify the nuisance estimation.

Semi-competing risks data have also been analyzed using alternative approaches, such as the win ratio and event-specific win ratio methods \citep{yang2022event}. In the causal inference literature, recent work has examined these data under mediation frameworks in which the non-terminal event is treated as a mediator of the treatment effect on the terminal event \citep{huang2021causal,xu2022bayesian}. While such approaches offer useful perspectives, they target different causal quantities. 
In contrast, our  method targets the total causal effects of treatment across all possible event pathways, and provides 
estimates of  treatment-specific risks as clinically meaningful  overall assessments of treatment effects.

In our analysis of the HAAS data, we have considered time since HAAS enrollment as the time scale, while alternative time scales such as age at event may also be of interest. On this time scale, left truncation occurs because participants entered the study at various ages, and individuals who had experienced events (e.g., MI or death) prior to their study entry were not included in the data, leading to selection bias.
We are currently working on a separate project to extend the proposed 
framework to semi-competing risks settings with covariate-dependent left truncation and right censoring, leveraging recent doubly robust methods for handling covariate-dependent left truncation \citep{wang2024doubly,wang2025liberatingframeworktruncationcensoring}.

\section*{Acknowledgments}
This research used services provided by the Open Science Grid (OSG) Consortium, 
supported by the National Science Foundation awards \#2030508 and \#1836650.
\appendix
\section{Appendix}

\subsection{Estimators of $\Lambda_j^a(t)$, $j=1,2,3,a=0,1$}
In the main paper, we presented the expressions for $\hat{\Lambda}_1^1(t)$ and $\hat{\Lambda}_3^1(t)$, which involve functions of $Q_{1i}^1, Q_{2i}^1, Q_{3i}^1, m_1, m_3, N_{1i}, N_{3i}, Y_{1i}, Y_{3i}$. Here we provide the expressions for $\hat{\Lambda}_1^0(t)$, $\hat{\Lambda}_2^1(t)$, $\hat{\Lambda}_2^0(t)$ and $\hat{\Lambda}_3^0(t)$, which were derived in a similar manner. 
Let 
\begin{align*}
Q_{1i}^0(t) &= \frac{1-A_i}{\{1-\hat \pi(Z_i)\} \hat S_c(t | Z_i,A=0)},\\
Q_{2i}^0(t) &= 1-\frac{1-A_i}{1-\hat \pi(Z_i)}+\frac{1-A_i}{1-\hat \pi(Z_i)} m_1(\delta_{2i},X_{1i},  X_{2i}, 0, Z_i, t),\\
Q_{3i}^0(t) &= 1-\frac{1-A_i}{1-\hat \pi(Z_i)}+\frac{1-A_i}{1-\hat \pi(Z_i)} m_3(\delta_{2i}, X_{2i}, 0, Z_i, t),\\
N_{2i}(t) &= \mathbbm{1}(X_{2i}\leqslant t, \delta_{1i}=0, \delta_{2i}=1).
\end{align*}

We have 
\begin{align*}
\hat{\Lambda}_1^0(t) &= \int_0^t\frac{
\sum_{i=1}^n Q_{1i}^0(u)dN_{1i}(u)+ Q_{2i}^0(u) \hat S(u | A=0, Z_i)d \hat \Lambda_1(u | A=0, Z_i) }{
\sum_{i=1}^n Q_{1i}^0(u)Y_{1i}(u)+ Q_{2i}^0(u)\hat S(u | A=0, Z_i)},\\
\hat{\Lambda}_2^a(t) &=\int_0^t \frac{
\sum_{i=1}^n Q_{1i}^a(u)dN_{2i}(u)+ Q_{2i}^a(u) \hat S(u | A=a, Z_i)d \hat \Lambda_2(u | A=a, Z_i) }{
\sum_{i=1}^n Q_{1i}^a(u)Y_{1i}(u)+ Q_{2i}^a(u)\hat S(u | A=a, Z_i)},\quad a = 0,1,\\
\hat{\Lambda}_3^0(t)&=\int_0^t\frac{
\sum_{i=1}^n Q_{1i}^0(u)dN_{3i}(u)+ Q_{3i}^0(u)  \{-d\hat S_2(u|A=0, Z_i)-\hat S(u|A=0,Z_i)d\hat \Lambda_2(u|A=0, Z_i)\} }{
\sum_{i=1}^n Q_{1i}^0(u)Y_{3i}(u)+ Q_{3i}^0(u)\{\hat S_2(u|A=0,Z_i)-\hat S(u|A=0,Z_i)\}}.
\end{align*}

\subsection{Expression of $\mathrm{CIF}_3^a(t|t_1)$}\label{append:CIF3}
\begin{align*}
\text{For } t>t_1: \mathrm{CIF}_3^a(t|t_1) &= 1-\exp\{-\Lambda_3^a(t|t_1)\} \\
    &= 1-\exp\left\{-\int_{0}^t\lambda_3^a(u|t_1)du\right\} \\
    &= 1-\exp\left\{-\int_{t_1}^t\lambda_3^a(u)du\right\} \quad \text{[Markov condition]}\\
    &= 1-\exp\left\{\Lambda_3^a(t_1)-\Lambda_3^a(t)\right\}.
\end{align*}


\bibliographystyle{apalike} 

\begin{singlespace}
\bibliography{references}

@article{andersen1991non,
  title={Non-and semi-parametric estimation of transition probabilities from censored observation of a non-homogeneous Markov process},
  author={Andersen, Per Kragh and Hansen, Lars Sommer and Keiding, Niels},
  journal={Scandinavian Journal of Statistics},
  pages={153--167},
  year={1991},
  publisher={JSTOR}
}

@article{andersen2022inference,
  title={Inference for transition probabilities in non-Markov multi-state models},
  author={Andersen, Per Kragh and Wandall, Eva Nina Sparre and Pohar Perme, Maja},
  journal={Lifetime Data Analysis},
  volume={28},
  number={4},
  pages={585--604},
  year={2022},
  publisher={Springer}
}

@article{bai2013doubly,
  title={Doubly-robust estimators of treatment-specific survival distributions in observational studies with stratified sampling},
  author={Bai, Xiaofei and Tsiatis, Anastasios A and O'Brien, Sean M},
  journal={Biometrics},
  volume={69},
  number={4},
  pages={830--839},
  year={2013},
  publisher={Oxford University Press}
}

@article{bai2017optimal,
  title={Optimal treatment regimes for survival endpoints using a locally-efficient doubly-robust estimator from a classification perspective},
  author={Bai, Xiaofei and Tsiatis, Anastasios A and Lu, Wenbin and Song, Rui},
  journal={Lifetime Data Analysis},
  volume={23},
  number={4},
  pages={585--604},
  year={2017},
  publisher={Springer}
}

@article{chernozhukov2018double,
    author = {Chernozhukov, Victor and Chetverikov, Denis and Demirer, Mert and Duflo, Esther and Hansen, Christian and Newey, Whitney and Robins, James},
    title = {Double/debiased machine learning for treatment and structural parameters},
    journal = {The Econometrics Journal},
    volume = {21},
    number = {1},
    pages = {C1-C68},
    year = {2018},
    month = {02},
    issn = {1368-4221},
    doi = {10.1111/ectj.12097},
    url = {https://doi.org/10.1111/ectj.12097},
    eprint = {https://academic.oup.com/ectj/article-pdf/21/1/C1/27684918/ectj00c1.pdf},
}

@article{dukes2019doubly,
  title={On doubly robust estimation of the hazard difference},
  author={Dukes, Oliver and Martinussen, Torben and Tchetgen Tchetgen, Eric J and Vansteelandt, Stijn},
  journal={Biometrics},
  volume={75},
  number={1},
  pages={100--109},
  year={2019},
  publisher={Wiley Online Library}
}

@article{fine2001semi,
  title={On semi-competing risks data},
  author={Fine, Jason P and Jiang, Hongyu and Chappell, Rick},
  journal={Biometrika},
  volume={88},
  number={4},
  pages={907--919},
  year={2001},
  publisher={Biometrika Trust}
}

@inproceedings{gill1997coarsening,
  title={Coarsening at random: Characterizations, conjectures, counter-examples},
  author={Gill, Richard D and Van Der Laan, Mark J and Robins, James M},
  booktitle={Proceedings of the First Seattle Symposium in Biostatistics: Survival Analysis},
  pages={255--294},
  year={1997},
  organization={Springer}
}

@article{heitjan1991ignorability,
  title={Ignorability and coarse data},
  author={Heitjan, Daniel F and Rubin, Donald B},
  journal={The Annals of Statistics},
  pages={2244--2253},
  year={1991},
  publisher={JSTOR}
}

@book{hernan2020causal,
  title     = {Causal Inference: What If},
  author    = {Hernán, Miguel A. and Robins, James M.},
  year      = {2020},
  publisher = {Chapman and Hall/CRC}
}

@article{hou2023treatment,
  title={Treatment effect estimation under additive hazards models with high-dimensional confounding},
  author={Hou, Jue and Bradic, Jelena and Xu, Ronghui},
  journal={Journal of the American Statistical Association},
  volume={118},
  number={541},
  pages={327--342},
  year={2023},
  publisher={Taylor \& Francis}
}

@article{huang2021causal,
  title={Causal mediation of semicompeting risks},
  author={Huang, Yen-Tsung},
  journal={Biometrics},
  volume={77},
  number={4},
  pages={1143--1154},
  year={2021},
  publisher={Wiley Online Library}
}

@article{Ishwaran2008random,
author = {Hemant Ishwaran and Udaya B. Kogalur and Eugene H. Blackstone and Michael S. Lauer},
title = {{Random survival forests}},
volume = {2},
journal = {The Annals of Applied Statistics},
number = {3},
publisher = {Institute of Mathematical Statistics},
pages = {841--860},
year = {2008},
doi = {10.1214/08-AOAS169},
URL = {https://doi.org/10.1214/08-AOAS169}
}

@article{Ishwaran2014random,
  title={Random survival forests for competing risks},
  author={Ishwaran, Hemant and Gerds, Thomas A and Kogalur, Udaya B and Moore, Richard D and Gange, Stephen J and Lau, Bryan M},
  journal={Biostatistics},
  volume={15},
  number={4},
  pages={757--773},
  year={2014},
  publisher={Oxford University Press}
}

@book{klein2006survival,
  title={Survival analysis: techniques for censored and truncated data},
  author={Klein, John P and Moeschberger, Melvin L},
  year={2006},
  publisher={Springer Science \& Business Media}
}

@article{lin2022doubly,
  title={Doubly-robust estimator of the difference in restricted mean times lost with competing risks data},
  author={Lin, Jingyi and Trinquart, Ludovic},
  journal={Statistical Methods in Medical Research},
  volume={31},
  number={10},
  pages={1881--1903},
  year={2022},
  publisher={SAGE Publications Sage UK: London, England}
}

@article{lok2018estimation,
  title={Estimation of the cumulative incidence function under multiple dependent and independent censoring mechanisms},
  author={Lok, Judith J and Yang, Shu and Sharkey, Brian and Hughes, Michael D},
  journal={Lifetime Data Analysis},
  volume={24},
  number={2},
  pages={201--223},
  year={2018},
  publisher={Springer}
}

@article{luo2025doubly,
  title={Doubly robust estimation under a possibly misspecified marginal structural Cox model},
  author={Luo, Jiyu and Rava, Denise and Bradic, Jelena and Xu, Ronghui},
  journal={Biometrika},
  volume={112},
  number={1},
  pages={asae065},
  year={2025},
  publisher={Oxford University Press}
}

@article{meira2019estimation,
  title={Estimation in the progressive illness-death model: A nonexhaustive review},
  author={Meira-Machado, Lu{\'\i}s and Sestelo, Marta},
  journal={Biometrical Journal},
  volume={61},
  number={2},
  pages={245--263},
  year={2019},
  publisher={Wiley Online Library}
}

@article{p2012honolulu,
  title={The Honolulu-Asia Aging Study: epidemiologic and neuropathologic research on cognitive impairment},
  author={Gelber, Rebecca P. and Launer, Lenore J. and White, Lon R.},
  journal={Current Alzheimer Research},
  volume={9},
  number={6},
  pages={664--672},
  year={2012},
  publisher={Bentham Science Publishers}
}

@article{putter2018non,
  title={Non-parametric estimation of transition probabilities in non-Markov multi-state models: the landmark Aalen--Johansen estimator},
  author={Putter, Hein and Spitoni, Cristian},
  journal={Statistical Methods in Medical Research},
  volume={27},
  number={7},
  pages={2081--2092},
  year={2018},
  publisher={SAGE Publications Sage UK: London, England}
}

@article{rava2023doubly,
  title={Doubly robust estimation of the hazard difference for competing risks data},
  author={Rava, Denise and Xu, Ronghui},
  journal={Statistics in Medicine},
  volume={42},
  number={6},
  pages={799--814},
  year={2023},
  publisher={Wiley Online Library}
}

@book{ridgeway2022toolkit,
  title={Toolkit for weighting and analysis of nonequivalent groups: a tutorial for the R TWANG package},
  author={Ridgeway, Greg and McCaffrey, Daniel F and Morral, Andrew R and Cefalu, Matthew and Burgette, Lane F and Pane, Joseph D and Griffin, Beth Ann},
  year={2022},
  publisher={Rand Santa Monica, Calif}
}

@inproceedings{robins1993information,
  author       = {Robins, James M.},
  title        = {Information recovery and bias adjustment in proportional hazards regression analysis of randomized trials using surrogate markers},
  booktitle    = {Proceedings of the Biopharmaceutical Section, American Statistical Association},
  pages        = {24--33},
  year         = {1993},
  address      = {Washington, D.C.},
  publisher    = {American Statistical Association}
}

@incollection{robins2008higher,
  title={Higher order influence functions and minimax estimation of nonlinear functionals},
  author={Robins, James and Li, Lingling and Tchetgen, Eric and van der Vaart, Aad and others},
  booktitle={Probability and statistics: essays in honor of David A. Freedman},
  volume={2},
  pages={335--422},
  year={2008},
  publisher={Institute of Mathematical Statistics}
}

@article{Rotnitzky2005inverse,
  title={Inverse Probability Weighting in Survival Analysis},
  author={Rotnitzky, Andrea and Robins, James M.},
  journal={Encyclopedia of Biostatistics},
  pages={},
  year={2005},
  publisher={John Wiley \& Sons, Ltd}
}

@article{rotnitzky2021characterization,
  title={Characterization of parameters with a mixed bias property},
  author={Rotnitzky, Andrea and Smucler, Ezequiel and Robins, James M},
  journal={Biometrika},
  volume={108},
  number={1},
  pages={231--238},
  year={2021},
  publisher={Oxford University Press}
}

@article{rubin1974estimating,
  title={Estimating causal effects of treatments in randomized and nonrandomized studies.},
  author={Rubin, Donald B},
  journal={Journal of Educational Psychology},
  volume={66},
  number={5},
  pages={688--701},
  year={1974},
  publisher={American Psychological Association}
}

@article{sjolander2017doubly,
  title={Doubly robust estimation of attributable fractions in survival analysis},
  author={Sj{\"o}lander, Arvid and Vansteelandt, Stijn},
  journal={Statistical Methods in Medical Research},
  volume={26},
  number={2},
  pages={948--969},
  year={2017},
  publisher={SAGE Publications Sage UK: London, England}
}

@book{tsiatis2006semiparametric,
  title={Semiparametric theory and missing data},
  author={Tsiatis, Anastasios A},
  year={2006},
  publisher={Springer}
}

@book{van2003unified,
  title={Unified methods for censored longitudinal data and causality},
  author={Van der Laan, Mark J and Robins, James M},
  volume={5},
  year={2003},
  publisher={Springer}
}

@article{van2025doubly,
  title={Doubly robust estimation of marginal cumulative incidence curves for competing risk analysis},
  author={van Hage, Patrick and le Cessie, Saskia and van Maaren, Marissa C and Putter, Hein and van Geloven, Nan},
  journal={Statistics in Medicine},
  volume={44},
  number={18-19},
  pages={e70066},
  year={2025},
  publisher={Wiley Online Library}
}

@article{wang2024doubly,
  title={Doubly robust estimation under covariate-induced dependent left truncation},
  author={Wang, Yuyao and Ying, Andrew and Xu, Ronghui},
  journal={Biometrika},
  volume={111},
  number={3},
  pages={789--808},
  year={2024},
  publisher={Oxford University Press}
}

@article{wang2025liberatingframeworktruncationcensoring,
  title={A Liberating Framework from Truncation and Censoring, with Application to Learning Treatment Effects},
  author={Wang, Yuyao and Ying, Andrew and Xu, Ronghui},
  journal={arXiv:2411.18879},
  year={2025}
}

@article{xu2010statistical,
  title={Statistical analysis of illness--death processes and semicompeting risks data},
  author={Xu, Jinfeng and Kalbfleisch, John D and Tai, Beechoo},
  journal={Biometrics},
  volume={66},
  number={3},
  pages={716--725},
  year={2010},
  publisher={Wiley Online Library}
}

@article{xu2022bayesian,
  title={A Bayesian nonparametric approach for evaluating the causal effect of treatment in randomized trials with semi-competing risks},
  author={Xu, Yanxun and Scharfstein, Daniel and M{\"u}ller, Peter and Daniels, Michael},
  journal={Biostatistics},
  volume={23},
  number={1},
  pages={34--49},
  year={2022},
  publisher={Oxford University Press}
}

@article{yang2022event,
  title={Event-specific win ratios for inference with terminal and non-terminal events},
  author={Yang, Song and Troendle, James and Pak, Daewoo and Leifer, Eric},
  journal={Statistics in Medicine},
  volume={41},
  number={7},
  pages={1225--1241},
  year={2022},
  publisher={Wiley Online Library}
}

@article{zhang2012double,
  title={Double-robust semiparametric estimator for differences in restricted mean lifetimes in observational studies},
  author={Zhang, Min and Schaubel, Douglas E},
  journal={Biometrics},
  volume={68},
  number={4},
  pages={999--1009},
  year={2012},
  publisher={Oxford University Press}
}

@article{zhang2024marginal,
  title={Marginal structural illness-death models for semi-competing risks data},
  author={Zhang, Yiran and Ying, Andrew and Edland, Steve and White, Lon and Xu, Ronghui},
  journal={Statistics in Biosciences},
  volume={16},
  number={3},
  pages={668--692},
  year={2024},
  publisher={Springer}
}
\end{singlespace}
\clearpage

\begin{table}
\caption{Simulation results. Models in red are misspecified: for the event time model, it does not follow a Cox specification based on the data generation mechanism; for the censoring model, Cox1 is the correct model $C\sim A+Z_1+Z_2$ and Cox2 is the misspecified model $C\sim A+Z_1^2+Z_2^2$; for the PS model, logit1 is the correct model $A\sim Z_1+Z_2$ and logit2 is the misspecified model $A\sim Z_1^2+Z_2^2$.\label{table:sim}}
\centering
\begin{tabular}[t]{lllrlll}
\toprule
Estimand:truth & Method & Event Time/Censor-PS & Bias & SD & SE & CP(\%)\\
\midrule
$\text{CIF}_1^0(3)$: 0.358 & Causal- & RSF/(Cox1-gbm) & -0.001 & 0.033 & 0.033 & 94.9\\
 & AIPCW\textsuperscript{sc} & RSF/(Cox1-logit1) & 0.001 & 0.033 & 0.031 & 93.7\\
 &  & RSF/(RSF-gbm) & 0.000 & 0.032 & 0.033 & 95.2\\
 &  & \textcolor{red}{Cox}/(Cox1-logit1) & -0.001 & 0.034 & 0.033 & 94.7\\
 &  & \textcolor{red}{Cox}/(RSF-gbm) & -0.005 & 0.033 & 0.033 & 95.5\\
 &  & \textcolor{red}{Cox}/(RSF-logit1) & 0.001 & 0.034 & 0.032 & 94.3\\
 &  & \textcolor{red}{Cox}/(\textcolor{red}{Cox2}-\textcolor{red}{logit2}) & -0.033 & 0.032 & 0.031 & 80.2\\
 \cline{2-7}
 & IPW & -/(Cox1-logit1) & -0.003 & 0.034 & 0.032 & 94.7\\
 &  & -/(RSF-gbm) & -0.009 & 0.032 & 0.033 & 94.0\\
 &  & -/(RSF-logit1) & -0.002 & 0.033 & 0.032 & 94.2\\
 &  & -/(\textcolor{red}{Cox2}-\textcolor{red}{logit2}) & -0.050 & 0.031 & 0.030 & 61.5\\
  \cline{2-7}
 & Naive & - & -0.050 & 0.031 & 0.030 & 61.5\\
 & Full & - & 0.000 & 0.022 & 0.021 & 93.6\\
\addlinespace
$\text{CIF}_2^0(3)$: 0.322 & Causal- & RSF/(Cox1-gbm) & 0.005 & 0.029 & 0.029 & 94.6\\
 &  AIPCW\textsuperscript{sc} & RSF/(Cox1-logit1) & 0.002 & 0.029 & 0.028 & 93.8\\
 &  & RSF/(RSF-gbm) & 0.005 & 0.029 & 0.030 & 95.0\\
 &  & \textcolor{red}{Cox}/(Cox1-logit1) & 0.000 & 0.030 & 0.029 & 94.8\\
 &  & \textcolor{red}{Cox}/(RSF-gbm) & 0.007 & 0.030 & 0.030 & 94.0\\
 &  & \textcolor{red}{Cox}/(RSF-logit1) & 0.001 & 0.030 & 0.029 & 94.2\\
 &  & \textcolor{red}{Cox}/(\textcolor{red}{Cox2}-\textcolor{red}{logit2}) & 0.033 & 0.030 & 0.030 & 81.0\\
 \cline{2-7}
 & IPW & -/(Cox1-logit1) & -0.001 & 0.031 & 0.030 & 94.5\\
 &  & -/(RSF-gbm) & 0.006 & 0.030 & 0.030 & 94.7\\
 &  & -/(RSF-logit1) & -0.001 & 0.030 & 0.029 & 94.2\\
 &  & -/(\textcolor{red}{Cox2}-\textcolor{red}{logit2}) & 0.033 & 0.031 & 0.031 & 81.6\\
  \cline{2-7}
 & Naive & - & 0.031 & 0.031 & 0.031 & 84.1\\
 & Full & - & 0.001 & 0.021 & 0.021 & 95.0\\
\addlinespace
$\text{CIF}_3^0(3|0.5)$: 0.395 & Causal- & RSF/(Cox1-gbm) & 0.021 & 0.064 & 0.068 & 97.2\\
 & AIPCW\textsuperscript{sc} & RSF/(Cox1-logit1) & 0.020 & 0.063 & 0.062 & 95.0\\
 &  & RSF/(RSF-gbm) & 0.018 & 0.064 & 0.069 & 97.2\\
&  & \textcolor{red}{Cox}/(Cox1-logit1) & 0.005 & 0.062 & 0.062 & 95.6\\
 &  & \textcolor{red}{Cox}/(RSF-gbm) & 0.018 & 0.064 & 0.066 & 95.8\\
 &  & \textcolor{red}{Cox}/(RSF-logit1) & 0.006 & 0.063 & 0.062 & 95.1\\
 &  & \textcolor{red}{Cox}/(\textcolor{red}{Cox2}-\textcolor{red}{logit2}) & 0.059 & 0.066 & 0.065 & 86.8\\
 \cline{2-7}
 & IPW & -/(Cox1-logit1) & -0.003 & 0.069 & 0.070 & 96.0\\
 &  & -/(RSF-gbm) & 0.019 & 0.069 & 0.072 & 95.3\\
 &  & -/(RSF-logit1) & -0.001 & 0.068 & 0.067 & 95.6\\
 &  & -/(\textcolor{red}{Cox2}-\textcolor{red}{logit2}) & 0.074 & 0.075 & 0.076 & 86.5\\
  \cline{2-7}
 & Naive & - & 0.077 & 0.076 & 0.079 & 88.2\\
 & Full & - & -0.005 & 0.073 & 0.075 & 96.0\\
\bottomrule
\end{tabular}
\end{table}

\begin{table}
\caption{Event counts by exposure status in the HAAS data\label{table:haas_event}. MI = moderate impairment. 
}
\centering
\begin{tabular}[t]{lccc}
\toprule
 & \makecell{Overall  \\N = 1,881} & \makecell{Non-heavy drinking  \\N = 1,450} & \makecell{Heavy drinking  \\N = 431}\\
\midrule
Event Type &  &  & \\
~Censored before MI or death & 357 (19\%) & 280 (19\%) & 77 (18\%)\\
~MI then censor & 261 (14\%) & 210 (14\%) & 51 (12\%)\\
~Death without MI & 637 (34\%) & 501 (35\%) & 136 (32\%)\\
~MI then death & 626 (33\%) & 459 (32\%) & 167 (39\%)\\
\bottomrule
\end{tabular}
\end{table}




\begin{figure}[htbp]
    \centering
    \begin{subfigure}[b]{0.47\textwidth}
        \centering
        \includegraphics[width=0.9\textwidth]{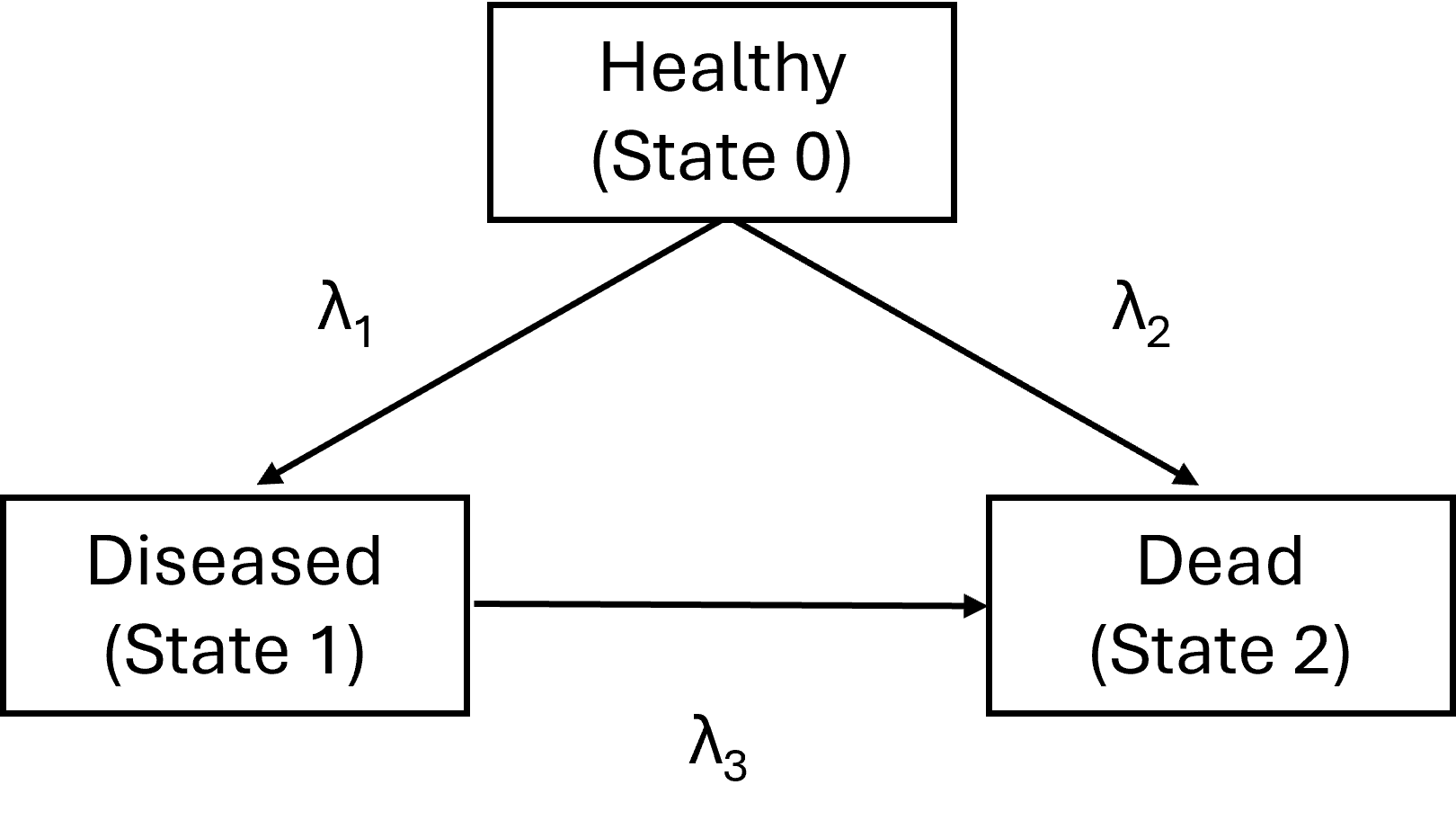}
        \caption{Illness-death model representation where $\lambda_1$, $\lambda_2$, and $\lambda_3$ denote the three transition rates}
        \label{fig:illness_death}
    \end{subfigure}
    \hfill
    \begin{subfigure}[b]{0.45\textwidth}
        \centering
        \includegraphics[width=0.8\textwidth]{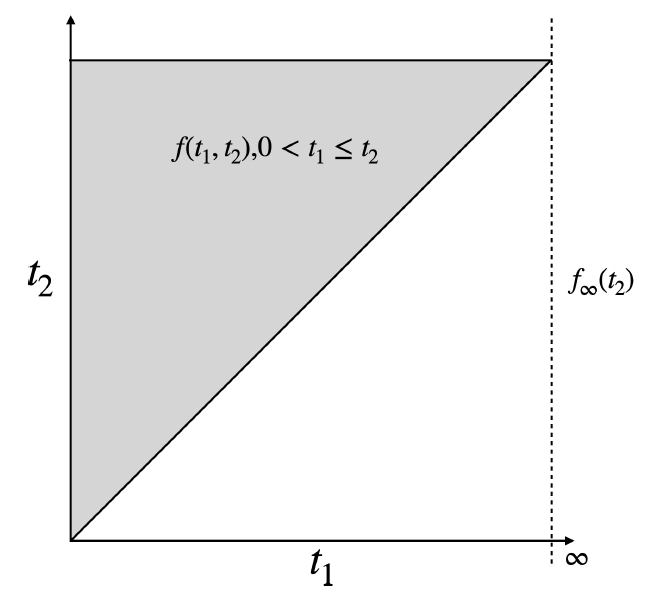}
        \caption{Joint density of non-terminal event time ($T_1$) and terminal event time ($T_2$)}
        \label{fig:joint_density}
    \end{subfigure}
    \caption{Semi-competing risks data}
\end{figure}

\begin{figure}[htbp]
    \centering
    \includegraphics[width=0.8\textwidth]{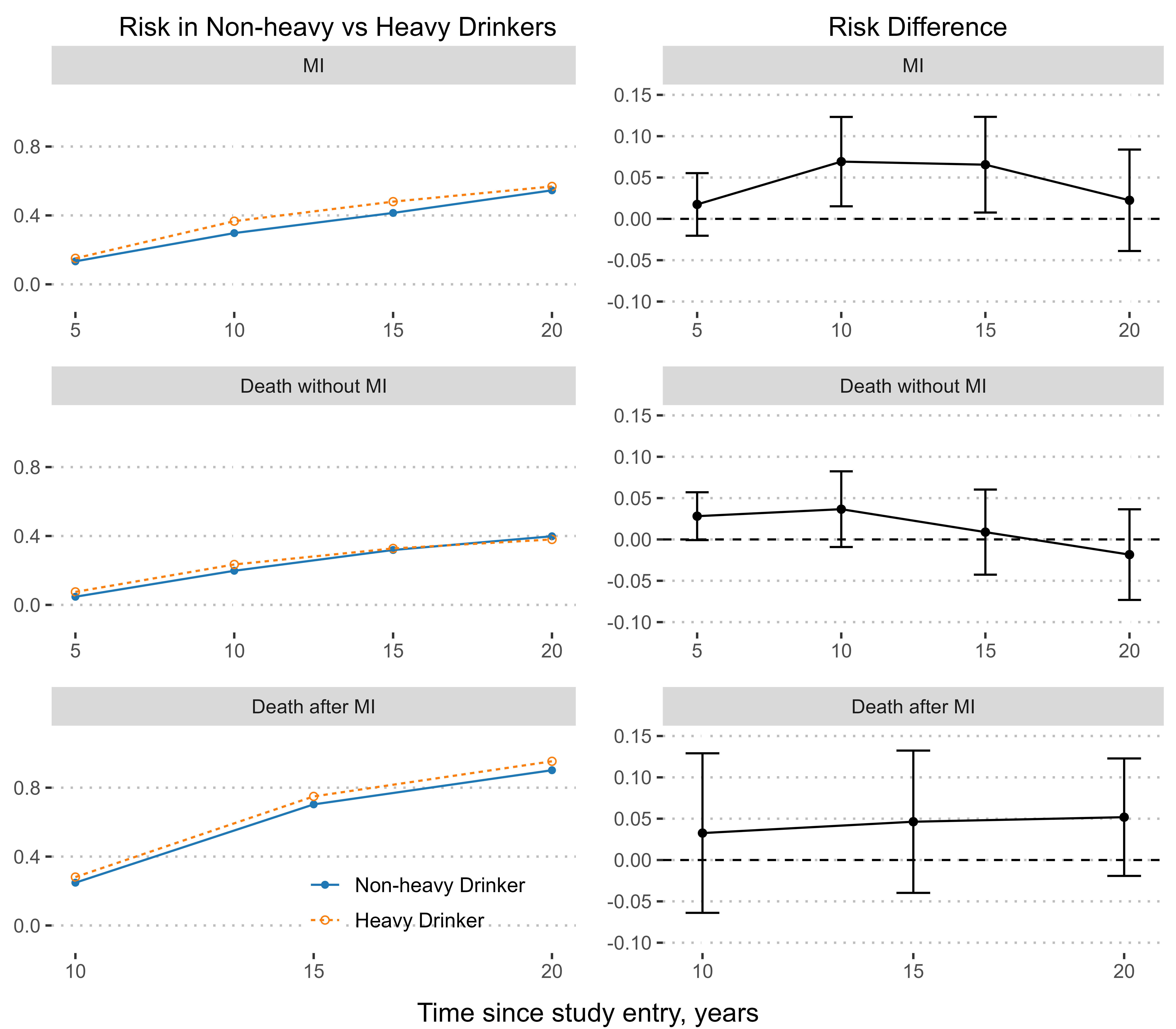}
        \caption{Estimated treatment-specific risks and risk differences for moderate impairment (MI), death without MI, and death after MI by 8 years. Error bars represent  95\% confidence intervals.}
\label{fig:haas}
\end{figure}

\clearpage
\pagenumbering{arabic}
\setcounter{page}{1}
\begin{singlespace}
\begin{center}
{\large\bf  Doubly Robust Estimation under Covariate Dependent Censoring in Semi-Competing Risks Data}\\[1.2em]
\end{center}
\end{singlespace}
\bigskip

\section{Supplementary Material}
\setcounter{figure}{0}
\setcounter{table}{0}
\renewcommand{\thefigure}{S\arabic{figure}}
\renewcommand{\thetable}{S\arabic{table}}
\localtableofcontents
\newpage 
\subsection{Derivation of AIPCW\textsuperscript{sc} via monotone coarsening} \label{append:derive_AIPCW}

\begin{lemma}\label{lemma:monotone}
Censoring in semi-competing risks data is monotone coarsening.
\end{lemma}

\begin{proof}[Proof of Lemma \ref{lemma:monotone}]
We need to show that for any $0<r<s\leqslant \infty$, there exists a fixed function $f_{r,s}$ such that $G_r(W) = f_{r,s}(G_s(W))$.
Define $f_{r,\infty}$ on arguments $g = (t^u, a, d_1, d_2, z)$ by
\[
f_{r,\infty}(g) = \{\mathbbm{1}(t^u\geqslant r),\, z\},
\]
and, for $r<s<\infty$, define $f_{r,s}$ on arguments $g = (i, z)$ by
\[
f_{r,s}(g) = (i,\, z).
\]
Consider the following two cases:
\begin{enumerate}
    \item If $s=\infty$, then $G_s(W) = W = (\Tearly, T_1\wedge\Tearly, \deltaearlyone, \deltaearlytwo, Z)$ and
    \[
    f_{r,\infty}(G_\infty(W)) = \{\mathbbm{1}(\Tearly\geqslant r), Z\} = G_r(W).
    \]

    \item If $r<s<\infty$, then on $\{R=s\}$ we have $C=s\leqslant\Tearly$, so $\Tearly\geqslant s$ and $G_s(W) = (1, Z)$. Since $r<s$, on this event we also have $\Tearly\geqslant r$, so that $G_r(W) = (1, Z)$. Therefore
    \[
    f_{r,s}(G_s(W)) = (1, Z) = G_r(W).
    \]
\end{enumerate}
Combining the two cases, we have $G_r(W) = f_{r,s}(G_s(W))$, establishing monotone coarsening.
\end{proof}

\begin{lemma}\label{lemma:CAR}
Assumption \ref{ass:aipcw_censoring} 
implies coarsening at random (CAR).
\end{lemma}

\begin{proof}[Proof of Lemma \ref{lemma:CAR}]
Let $\Hearly = (\Tearly, T_1\wedge \Tearly, \deltaearlyone, \deltaearlytwo)$. Then $W =(\Hearly, Z) $.
Since $\Hearly$ is a function of $(T_1, T_2, t)$, by Assumption \ref{ass:aipcw_censoring}, we have $C\indep \Hearly \mid  Z$.
Recall that CAR is defined as $P(R=r| W) = w(r,G_r(W))$ for some $w$.
Let $f_c$ and $S_c$ denote the conditional density and survival functions of $C$ given $Z$.
Define
\[w(r,g) = \begin{cases} S_c(t^u|Z=z), &  \quad r=\infty, \ g = (h^u, z) \text{ with } h^u = (t^u, a, d_1, d_2),\\ i\,f_c(r|Z=z), & \quad r<\infty, \ g = (i, z). \end{cases} \]
Now we verify that $P(R=r| W) = w(r,G_r(W))$.
First, for \(r=\infty\), since $G_\infty(W)=W=(\Hearly, Z)$,
\begin{align*}
P(R=\infty| W) & =P(\Tearly< C \mid  \Hearly, Z) \\
&=S_c(\Tearly | Z) \quad  [\text{by }C\indep \Hearly \mid  Z ], \\
w(\infty,G_\infty(W)) &= w(\infty,(\Hearly, Z))\\
&=S_c(\Tearly| Z).
\end{align*}
So $P(R=\infty| W) =w(\infty,G_\infty(W))$.
Next, for \(r<\infty\), the conditional density of \(R\) given \(W\) is
\begin{align*}
 P(R=r| W) &\equiv  \lim_{\Delta \to 0^+}  P(R \in [r, r+\Delta)|W)/\Delta\\
&=\lim_{\Delta \to 0^+} P(\Tearly \geqslant C, C \in [r, r+\Delta) \mid  \Hearly, Z)/\Delta \\
&=\mathbbm{1}(\Tearly\geqslant r) \lim_{\Delta \to 0^+}P( C \in [r, r+\Delta) \mid  \Hearly, Z)/\Delta \\
&=\mathbbm{1}(\Tearly\geqslant r)f_c(r \mid \Hearly, Z) \\
&=\mathbbm{1}(\Tearly\geqslant r)f_c(r | Z) \quad   [\text{by }C\indep \Hearly \mid  Z ],\\
w(r,G_r(W)) &= w(r,\{\mathbbm{1}(\Tearly\geqslant r), Z\})\\
&=\mathbbm{1}(\Tearly\geqslant r)f_c(r| Z).
\end{align*}
Hence, $P(R=r| W)=w(r,G_r(W))$ for $r<\infty$.
Combining the two cases, we have $P(R=r| W)=w(r,G_r(W))$,  establishing CAR.
\end{proof}

\bigskip
Now we derive the AIPCW\textsuperscript{sc} estimating function  in  \eqref{eq:aipcw}. 
In the following, we compute each of the quantities $P(R = \infty|W)$, $dM_R(r)$ and $K_r$ for our censoring problem. 
We have 
\begin{align}
P(R=\infty | W) & =P(\Tearly< C \mid  \Hearly, Z)\nonumber \\
& =S_c(\Tearly | Z) \quad [\text{by }C\indep \Hearly \mid Z].
\end{align}
It remains to compute $dM_R(r)$ and $K_r(W)$ for $r<\infty$.

\begin{enumerate}
  \item $dM_R(r)$: 
  Recall that $R=r < \infty$ when  $C=r \leqslant \Tearly$, while $R=\infty$ when  $\Tearly< C$. 
  Thus,  $\{R \geqslant r\}= \{r \leqslant C \leqslant \Tearly\} \cup \{\Tearly< C\} $, and $ \{ R \in [r, r+\Delta) \} = \{ C \in [r, r+\Delta), \Tearly\geqslant C\}$. 
  
Therefore 
\begin{align*}
\lambda_{r}(W)&\stackrel{\text{def}}{=} \lim_{\Delta \to 0^+}  {P(R \in [r, r+\Delta) | R\geqslant r, W) }/{\Delta}\\
&=\lim_{\Delta \to 0^+}  {P(C \in [r, r+\Delta), \Tearly \geqslant C \mid \{r \leqslant C \leqslant \Tearly\} \cup  \{\Tearly< C\} , \Hearly, Z ) }/{\Delta} \\
&=\mathbbm{1}(\Tearly\geqslant r)\lim_{\Delta \to 0^+}  {P(C \in [r, r+\Delta) \mid C\geqslant r, \Hearly, Z) }/{\Delta}\\
&=\mathbbm{1}(\Tearly\geqslant r)\lim_{\Delta \to 0^+}  {P(C \in [r, r+\Delta) \mid C\geqslant r, Z) }/{\Delta} \quad [\text{by }C\indep \Hearly\mid Z]\\
& =\mathbbm{1}(\Tearly\geqslant r)\lambda_{c}(r | Z).
\end{align*}
Note that 
\[
\mathbbm{1}(R=r)=\mathbbm{1}(C=r, \Tearly \geqslant r)= \mathbbm{1}( \Xearly\geqslant r) d N_{c}(r) ,
\]
and
\begin{align}
\lambda_{r}(W)\mathbbm{1}(R\geqslant r)&=\lambda_{c}(r | Z) \mathbbm{1}(\Tearly\geqslant r) \mathbbm{1}(R\geqslant r) \label{eq:lambda_r_ind}\\
& =\lambda_{c}(r | Z) Y_{c}(r) \mathbbm{1}(\Xearly \geqslant r). \nonumber
\end{align}
Therefore
$$
d M_R(r) = \mathbbm{1}( \Xearly\geqslant r) d M_{c}(r). 
$$

\item $K_r$: 
\begin{align}
K_r(W) & \stackrel{\text {def}}{=} P(R>r | W) \nonumber\\
& =P(\{\Tearly< C\} \cup\{r <C \leqslant \Tearly \}\mid \Hearly, Z) \nonumber\\
& = \begin{cases}P(C > r \mid \Hearly, Z) & \text { if } \Tearly> r \nonumber\\
P(C>\Tearly \mid \Hearly, Z) & \text { if } \Tearly\leqslant r\end{cases} \\
& = S_c(\Tearly \wedge r | Z), 
\label{eq:K_r}
\end{align} 
where $ \Tearly \wedge r = \min(\Tearly, r) $.

\end{enumerate}
Finally putting the above together we have
 the most efficient augmentation term 
\begin{align*}
& \int_{0}^{\infty} \frac{ \mathbbm{1}( \Xearly\geqslant r) d M_{c}(r)
}{S_c( \Tearly \wedge r|Z)} L(r, t ; \theta ; F ; Z) \\
= & \int_{0}^{\Xearly} \frac{L(r, t ; \theta ; F ; Z)}{S_c(r|Z)} d M_{c}(r) .
\end{align*}
Combining the IPW term and the augmentation term gives the AIPCW\textsuperscript{sc} formula presented in \eqref{eq:aipcw}.

\subsection{Derivation of the AIPTW estimating functions}\label{append:aiptw}

We focus on estimating $\theta^1$ first. Recall that  $ \fullestone$ is the full data estimating function. The IPTW estimating function for $\theta^1$ is 
$$
\Uiptw(t;\theta^1;\pi) = \frac{A}{\pi(Z)}\fullestone.
$$
By the method and conclusions described in \citet[Chapter~1, Page 87]{van2003unified}, let $\mathcal{T}$ denote the propensity score tangent space and $\Pi(\cdot|\mathcal{T})$ denote the projection operator onto $\mathcal{T}$, 
then we have: 
{\small
\begin{align}
\Pi(\Uiptw|\mathcal{T})&=E(\Uiptw|A,Z)-E(\Uiptw|Z) 
\nonumber\\
&=E\left\{\frac{A}{\pi(Z)} \fullestone\Big| A, Z\right\}-E\left\{\frac{A}{\pi(Z)} \fullestone\Big| Z\right\} \nonumber \\
& =\frac{A}{\pi(Z)} E\{\fullestone | A, Z\}-\frac{1}{\pi(Z)} E\{A\fullestone | Z\} \nonumber \\
& =\frac{A}{\pi(Z)} E\{\fullestone | Z\}-\frac{1}{\pi(Z)} E(A| Z)E\{\fullestone | Z\} \label{eq:aiptw_derive}\\
& =\left\{\frac{A}{\pi(Z)}-1\right\} E\{\fullestone | Z\}\nonumber,
\end{align}}
where \eqref{eq:aiptw_derive} follows by Assumption \ref{ass:causal_exchange}.
Thus the AIPTW estimating function is:
\begin{align}\label{eq:aiptw}
U_1^\text{aiptw}(t;\theta^1;\pi,F)&=\Uiptw - \Pi(\Uiptw|\mathcal{T})\nonumber  \\ 
&=\frac{A}{\pi(Z)}\fullestone +\left\{1-\frac{A}{\pi(Z)}\right\} E\{\fullestone | Z\}\nonumber.
\end{align}
Recall that $U^f(t; \theta^1)\equiv U^f(t;\theta^1; T_1,T_2, Z)$. Then, by Assumptions \ref{ass:causal_consistency} and \ref{ass:causal_exchange},
\begin{equation}\label{eq:aiptw1}
U_1^\text{aiptw}(t;\theta^1;\pi,F)=\frac{A}{\pi(Z)} U^f(t;\theta^1)+\left\{1-\frac{A}{\pi(Z)}\right\} E\{U^f(t;\theta^1) |A=1, Z\}.
\end{equation}
Similarly, the AIPTW estimating function for $\theta^0$ is: 
\begin{equation}
U_0^\text{aiptw}(t;\theta^0;\pi,F)=\frac{1-A}{1-\pi(Z)}U^f(t;\theta^0) +\left\{1-\frac{1-A}{1-\pi(Z)}\right\}E\{U^f(t;\theta^0) | A=0, Z\}.
\end{equation}

\subsection{Formulation of causal-AIPCW\textsuperscript{sc} via monotone coarsening}\label{append:causal_monotone}
\newcommand{\Tearlyone}{T^{u,1}}  
\newcommand{\Xearlyone}{X^{u,1}}  

We present the formulation for $\theta^1$; the formulation for $\theta^0$ is analogous. Let $\Tearlyone$ 
denote the earliest follow-up time at which the value of the estimating function
$\fullestone$ becomes fully observed; let  $W=(\Tearlyone, T_1^1\wedge \Tearlyone, \deltaearlyoneone, \deltaearlytwoone, A, Z)$ denote the data that we wish to observe, where $\deltaearlyoneone = \mathbbm{1}(T_1^1\leqslant \Tearlyone)$ and $\deltaearlytwoone = \mathbbm{1}(T_2^1\leqslant \Tearlyone)$; let $(R, G_R(W))$ denote the coarsened data  we actually observe. Here, $R$ is the coarsening variable such that when $R=r$, we observe a many-to-one function $G_r(W)$, instead of $W$ itself. By convention, $R=\infty$ denotes no coarsening. 
Specifically, we define $(R, G_R(W))$ in the following way: 

\begin{align}\label{eq:causal_monocoar}
& \left\{\begin{array}{l}
\text { if } \Tearlyone <  C^1, A=1, \text { then } R=\infty, G_{\infty}(W)=W; \\
\text { if } \Tearlyone\geqslant C^1=r>0, A=1, \text { then } R=r, G_{r}(W)= \{\mathbbm{1}(\Tearlyone\geqslant  r),A, Z\};\\
\text { if } A=0, \text { then } R=0, G_{0}(W)= \{A, Z\}.
\end{array} \right. 
\end{align}

\subsection{Identifiability and double robustness proofs}\label{append:DRproof}
\begin{proof}[Proof of Lemma \ref{lemma:DR}] 
In the following, we consider the two cases $S_c = S_c^o$ and $F =F^o$ separately.

a) We first consider the case with $S_c=S_c^o$. 
Recall that 
\begin{equation*}
U(t;\theta;S_c^o, F)= 
\frac{\mathbbm{1}(\Tearly< C) U^c(t ; \theta)}{S_c^o(\Xearly | Z)}
+\int_0^{\Xearly} \frac{L(r, t ; \theta ; F;Z)}{S_c^o(r | Z)} d M_c(r ; S_c^o).
\end{equation*}
For the first term above, since $ \Xearly = \min(\Tearly, C) $ we have
\begin{align*}
E\left\{\frac{\mathbbm{1}(\Tearly< C) U^c(t ; \theta)}{S_c^o(\Xearly| Z)}~\Big|~T_1, T_2, Z\right\}&= 
E\left\{\frac{\mathbbm{1}(\Tearly< C) U^c(t ; \theta)}{S_c^o(\Tearly| Z)}~\Big|~T_1, T_2, Z\right\}\\
&=\frac{U^c(t ; \theta)}{S_c^o(\Tearly | Z)}E\{\mathbbm{1}(C>\Tearly)~|~T_1, T_2, Z\} \\
&=\frac{U^c(t ; \theta)}{S_c^o(\Tearly | Z)}S_c^o(\Tearly | Z) \quad \text{[by Assumption \ref{ass:aipcw_censoring}]}\\
&=U^c(t ; \theta).
\end{align*} 
By the tower rule, we have 
\begin{equation}\label{eq:DRproof1_1}
E\left\{\frac{\mathbbm{1}(\Tearly< C) U^c(t ; \theta)}{S_c^o(\Xearly | Z)}\right\}=E\left[E\left\{\frac{\mathbbm{1}(\Tearly< C) U^c(t ; \theta)}{S_c^o(\Xearly | Z)}~\Big |~T_1,T_2,Z\right\}\right] =E\{U^c(t ; \theta)\}. 
\end{equation}
For the second term, we will show below that
\begin{equation}
\label{eq:DRproof1_2}
E\left\{\int_0^{\Xearly} \frac{L(r, t ; \theta ; F;Z)}{S_c^o(r | Z)} d M_c(r ; S_c^o)\right\}=0.
\end{equation}
Combining \eqref{eq:DRproof1_1} and \eqref{eq:DRproof1_2} then, we have $E\{U(t;\theta;S_c^o, F)\}=E\{U^c(t;\theta)\}$ for all $\theta$. 

The slight challenge in showing \eqref{eq:DRproof1_2} is to show that $\mathbbm{1}(\Xearly \geqslant r)$ is predictable with respect to 
the filtration of $M_c$ which is $\{\mathcal{F}_t\}_{t\geq0}=\sigma(Z, \mathbbm{1}(T_2\leqslant \mu), \mathbbm{1}(C\leqslant \nu):\mu \leqslant t, \nu \leqslant t)$.  
Instead consider $\tilde M_c(s;S_c^o) = \mathbbm{1} (C\leqslant s)- \int_0^s\mathbbm{1}( C\geqslant r)\lambda_c(r| Z)dr$. Recall that $M_c(t ; S_c)=N_{c}(t)-\int_{0}^{t} Y_{c}(r) \lambda_{c}(r| Z) dr$, where $N_{c}(t)=\mathbbm{1}( X_{2} \leqslant t, \delta_{2}=0)$, 
$Y_{c}(t)=\mathbbm{1}( X_{2} \geqslant t)$. For $r\in (0,\Xearly)$, we have 
$r < \Xearly\leqslant \Tearly\leqslant T_2$, so
\begin{align*}
    dM_c(r;S_c^o) 
    & = \mathbbm{1}(X_2 = r, \delta_2 = 0) - \mathbbm{1}( X_{2} \geqslant r)\lambda_c(r|Z) dr\\
    & = \mathbbm{1}(C = r, T_2>r) - \mathbbm{1}(C\geqslant r, T_2 \geqslant r)\lambda_c(r|Z) dr \\
    & = \mathbbm{1}(C = r) - \mathbbm{1}(C\geqslant r)\lambda_c(r|Z) dr \\
    & = d\tilde M_c(r;S_c^o).
\end{align*}
Therefore,
{\small
\begin{align}
E\left\{\int_0^{\Xearly} \frac{L(r, t ; \theta ; F;Z)}{S_c^o(r | Z)} d M_c(r ; S_c^o)\right\} 
&= 
E\left\{\int_0^{\Xearly} \frac{L(r, t ; \theta ; F;Z)}{S_c^o(r | Z)}  d \tilde M_c(r ; S_c^o)\right\} \nonumber \\
&=E\left\{\int_0^{\tau} \mathbbm{1}(\Xearly \geqslant  r)\frac{L(r, t ; \theta ; F;Z)}{S_c^o(r | Z)}\, d \tilde M_c(r ; S_c^o)\right\}. \label{eq:DRproof1_3}
\end{align}}

We now show that $\tilde M_c(s;S_c^o)$ is a martingale with respect to the filtration $\{\mathcal{H}_s\}_{s\geqslant0}=\sigma(Z, \mathbbm{1}(\Tearly\leqslant \mu), \mathbbm{1}(C\leqslant \nu):\mu,\nu\leqslant s)$.  
This is because $\Tearly$ is a function of $(T_1,T_2,t)$, Assumption \ref{ass:aipcw_censoring} implies $C\indep \Tearly\mid Z$. 
In addition, $\tilde M_c(s;S_c^o)$ is a martingale with respect to the filtration $\{\tilde{\mathcal H}_s\}_{s\geqslant0}=\sigma(Z,\mathbbm{1}(C\leqslant \nu):\nu\leqslant s)$. This together with $C\indep \Tearly\mid Z$ implies that $\tilde M_c(s;S_c^o)$ is also a martingale with respect to  $\{\mathcal{H}_s\}_{s\geqslant0}$. 
Finally, the integrand on the second line of \eqref{eq:DRproof1_3} is predictable with respect to $\{\mathcal{H}_s\}_{s\geqslant0}$, which implies that the second line of \eqref{eq:DRproof1_3} is zero, so \eqref{eq:DRproof1_2} follows.

\bigskip
b) We now consider the case with $F = F^o$. 
We first rewrite the expression of $U$, which will be useful in the later proof. Recall that $\Xearly = \min(\Tearly, C)$, $N_{c}(r)=\mathbbm{1}( X_2\leqslant r, \delta_{2}=0)$, $X_2 = \min(T_2, C)$ and $\delta_2 = \mathbbm{1}(X_2 = T_2)$.
Since 
\begin{align*}
\int_0^{\Xearly} \frac{d N_c(r)}{S_c(r | Z)}&=\frac{\mathbbm{1}(X_2\leqslant \Xearly,\delta_2=0)}{S_c(X_2 | Z)}
\\
&=\frac{\mathbbm{1}(C\leqslant \Xearly,C< T_2)}{S_c(C | Z)}=\frac{\mathbbm{1}(C\leqslant  \Tearly)}{S_c(\Xearly | Z)},
\end{align*}
$$\int_0^{\Xearly} \frac{\lambda_c(r| Z)}{S_c(r | Z)}dr = \int_0^{\Xearly} \frac{f_c(r| Z)}{S_c^2(r | Z)}dr=\frac{1}{S_c(\Xearly| Z)}-1,$$
we have 
\begin{align*}
\int_0^{\Xearly} \frac{d M_c(r; S_c)}{S_c(r | Z)}&=\int_0^{\Xearly} \frac{d N_c(r)}{S_c(r | Z)} -\int_0^{\Xearly} \frac{\lambda_c(r| Z)}{S_c(r | Z)}dr\\
&=1-\frac{\mathbbm{1}(C>\Tearly)}{S_c(\Xearly| Z)}.
\end{align*}
So 
$$\frac{\mathbbm{1}(C>\Tearly)}{S_c(\Xearly| Z)} =1-\int_0^{\Xearly} \frac{d M_c(r; S_c)}{S_c(r | Z)}.$$
Then,  
\begin{equation}
U(t;\theta;S_c, F^o) 
= U^c(t ; \theta)
-  \int_0^{\Xearly}\{U^c(t ; \theta) -L(r, t ; \theta ; F^o)\}\frac{d M_c(r ; S_c)}{S_c(r | Z)}. \label{eq:DRproof2_1}
\end{equation}
We will show that the second term in \eqref{eq:DRproof2_1} has expectation 0. 
We have 
\begin{align}
&E\left[\left. \int_0^{\Xearly} \{U^c(t ; \theta)-L(r, t ; \theta; F^o;Z)\}\frac{d M_c(r ; S_c)}{S_c(r | Z)} ~\right|~ C, Z\right] \nonumber\\
=&E\left[\left.\int_0^{\infty} \{U^c(t ; \theta)-L(r, t ; \theta ; F^o;Z)\}\mathbbm{1}( \Xearly \geqslant r)\frac{d \tilde M_c(r ; S_c)}{S_c(r | Z)} ~\right|~ C, Z\right] \label{eq:DRproof3}\\
=&\int_0^{\infty} E\left[\{U^c(t ; \theta)-L(r, t ; \theta ; F^o;Z)\}\frac{\mathbbm{1}( \Xearly \geqslant r)}{S_c(r | Z)} \Big |~C, Z\right]d\tilde M_c(r ; S_c), \nonumber
\end{align}
where \eqref{eq:DRproof3} holds because by definition $\Xearly\leqslant \Tearly\leqslant T_2$.
Since $\mathbbm{1}(\Xearly\geqslant r) = \mathbbm{1}(\Tearly\geqslant r)\mathbbm{1}(C\geqslant r)$, we have that the integrand in the above satisfies
\begin{align}
&E\left[\{U^c(t ; \theta)-L(r, t ; \theta ; F^o;Z)\}\frac{\mathbbm{1}(\Xearly\geqslant r)}{S_c(r | Z)} ~\Big |~C, Z\right] \nonumber\\
=&\frac{\mathbbm{1}( C\geqslant r)}{S_c(r | Z)}E\left[\{U^c(t ; \theta)-L(r, t ; \theta ; F^o;Z)\}\mathbbm{1}(\Tearly\geqslant r) \mid C, Z\right] \nonumber \\
=&\frac{\mathbbm{1}( C\geqslant r)}{S_c(r | Z)}[E\{U^c(t ; \theta)\mathbbm{1}(\Tearly\geqslant r)|C,Z \}
-E\{L(r, t ; \theta ; F^o;Z) \mathbbm{1}(\Tearly\geqslant r)  \mid C, Z\}] \nonumber\\
=&\frac{\mathbbm{1}( C\geqslant r)}{S_c(r | Z)}[E\{U^c(t ; \theta)\mathbbm{1}(\Tearly\geqslant r)|C,Z \}
-L(r, t ; \theta ; F^o;Z) E\{\mathbbm{1}(\Tearly\geqslant r)  \mid C, Z\}] \nonumber\\
=&\frac{\mathbbm{1}(C\geqslant r)}{S_c(r | Z)}[E\{U^c(t ;\theta)\mathbbm{1}(\Tearly\geqslant r)| Z\}-E\{U^c(t ; \theta)\mid \Tearly\geqslant r, Z\}P(\Tearly\geqslant r | Z)] \label{eq:DRproof2_2}\\
=&0. \nonumber
\end{align}
Equation \eqref{eq:DRproof2_2} holds 
because, since $\Tearly$ is a function of $(T_1, T_2, t)$, Assumption \ref{ass:aipcw_censoring} gives $ \Tearly \indep C\mid Z$, so the conditioning on $C$ can be dropped. The final equality follows because $E\{U^c(t;\theta)\mathbbm{1}(\Tearly\geqslant r)| Z\} = E\{U^c(t;\theta)\mid \Tearly\geqslant r, Z\}P(\Tearly\geqslant r| Z)$ by the definition of conditional expectation.
This, together with \eqref{eq:DRproof2_1}, implies $E\{U(t ; \theta;S_c, F^o)\} = E\{U^c(t ; \theta)\}$ for all $\theta$. 

Combining the two cases, we have proved the result and identifiability follows trivially under Assumption \ref{ass:esteq}.
\end{proof}

\vskip .2in
\begin{proof}[Proof of Lemma \ref{lemma:causalDR}]  \label{append:DRproof_causal}

By \cite{van2003unified} Lemma 1.9, it can be shown that under Assumptions \ref{ass:causal_sutva}, \ref{ass:causal_consistency}, \ref{ass:causal_exchange}, \ref{ass:causal_positive}, for $a=0,1$,  $E\{U_a^\text{aiptw}(t;\theta^a;\pi, F)\}=E\{\fullest\}$ for all $\theta^a$ if either $\pi = \pi^o$ or $F = F^o$.

Now we consider the following two cases.
a) $S_c=S_c^o$ and $\pi = \pi^o$. In this case, $E\{U_a^\text{aiptw}(t;\theta^a;\pi^o, F)\}=E\{\fullest\}$ for all $\theta^a$. Now we treat $U_a^\text{aiptw}(t;\theta^a;\pi^o, F)$ as the censoring-free data estimating function.
Recall that the causal-AIPCW\textsuperscript{sc} estimating function $U^a(t;\theta^a; S_c, \pi, F)$ is derived by applying our AIPCW\textsuperscript{sc} method to $U_a^\text{aiptw}$. Thus, under Assumptions \ref{ass:causal_positive} and \ref{ass:causal_censoring}, by Lemma \ref{lemma:DR}, $E\{U^a(t;\theta^a; S_c^o, \pi^o, F) \}=E\{U_a^\text{aiptw}(t;\theta^a;\pi^o, F)\}$ for all $\theta^a$. Thus $E\{U^a(t;\theta^a; S_c, \pi^o, F) \}=E\{\fullest\}$ for all $\theta^a$.
 b) $F=F^o$. In this case,  $E\{U_a^\text{aiptw}(t;\theta^a;\pi, F^o)\}=E\{\fullest\}$ for all $\theta^a$. Now we treat $U_a^\text{aiptw}(t;\theta^a;\pi, F^o)$ as the censoring-free data estimating function. Again under Assumptions \ref{ass:causal_positive} and \ref{ass:causal_censoring}, by Lemma \ref{lemma:DR}, $E\{U^a(t;\theta^a; S_c, \pi, F^o)\}=E\{U_a^\text{aiptw}(t;\theta^a;\pi, F^o)\}$ for all $\theta^a$. Thus, $E\{U^a(t;\theta^a; S_c, \pi, F^o)\}=E\{\fullest\}$  for all $\theta^a$. Combining the two cases, we have proved the result and  identifiability follows trivially under Assumption \ref{ass:causal_esteq}.  
\end{proof}

\subsection{Derivation of estimators for  $\Lambda_j^a(t)$, $j=1,2,3,a=0,1$}\label{append:derive_est}
We give detailed derivations of $\hat{\Lambda}_1^1(t)$ and $\hat{\Lambda}_3^1(t)$. The derivations of $\hat{\Lambda}_1^0(t)$, $\hat{\Lambda}_2^1(t)$, $\hat{\Lambda}_2^0(t)$ and $\hat{\Lambda}_3^0(t)$ are similar. 

Recall from \eqref{eq:c-aipcw} that the causal-AIPCW\textsuperscript{sc} estimating function for $\theta^1$ is: 
\begin{align}\label{eq:est_theta1}
U^1(t;\theta^1;S_c,\pi, F)&= 
\frac{A}{\pi(Z)}\frac{ \mathbbm{1}(\Tearly< C)} {S_c(\Xearly \mid A=1, Z)} U^f(t ; \theta^1) \nonumber\\
&\quad 
+\left\{1-\frac{A}{\pi(Z)} \right\} E\{U^f(t ; \theta^1) \mid A=1, Z\}\nonumber\\
+ & \frac{A}{\pi(Z)} \int_0^{\Xearly} \frac{E\{U^f(t;\theta^1) \mid \Tearly \geqslant r, A=1, Z\}}{S_c(r | A=1, Z)} d \causalMc(r; S_c),
\end{align}
where $\causalMc(t ; S_c)=\mathbbm{1}(X_2\leqslant t, \delta_2 = 0)-\int_0^t \mathbbm{1}(X_2 \geqslant u) \lambda_c(u | A,Z) du$.

To estimate $\Lambda_1^1(t)$, we consider  $\theta^1= \Lambda^1_1(t)$, $U^f(t; \theta^1)=dM_1^f(t;\theta^1)$, where $M_1^f(t;\theta^1) \equiv M_1^f(t;\theta^1;T_1, T_2) =  \mathbbm{1}(T \leqslant t, T_1\leqslant T_2)-\int_0^t\mathbbm{1}(T\geqslant u) d\Lambda_1^1(u)$,
and $\Tearly = \min(T, t).$ Now we plug these into \eqref{eq:est_theta1} above 
which consists of three terms, and we evaluate them separately. 

For the first term, we have 
\begin{equation*}
U^{f}(t;\theta^1) = dM_1^f(t;\theta^1)=\mathbbm{1}(T_1= t, T_1\leqslant T_2)-\mathbbm{1}(T\geqslant t) d\Lambda_1^1(t)
\end{equation*} 
and 
\begin{equation}\label{eq:lambda1_term1}
\frac{A}{\pi(Z)}\frac{ \mathbbm{1}(\Tearly< C)} {S_c(\Xearly | A=1, Z)} U^f(t ; \theta^1) = \frac{A\{\mathbbm{1}(X_1=t)\delta_1-\mathbbm{1}( X_1 \geqslant t)d\Lambda_1^1(t)\}}{\pi(Z)S_c(t| A=1,Z)}
\end{equation}
For the second term, we need to compute $E\{U^f(t;\theta^1) \mid A=1, Z\}$. For the third term, we need to  compute $E\{U^f(t;\theta^1) \mid \Tearly \geqslant r, A=1, Z\}$ when $t\geqslant r$, which reduces to  $E\{U^f(t;\theta^1) \mid A=1, Z\}$ by taking $r=0$. Thus, it suffices to focus on $E\{U^f(t;\theta^1) \mid  \Tearly \geqslant r, A=1, Z\}$. Let $N_1^f(t) = \mathbbm{1}(T \leqslant t, T_1\leqslant T_2)$ and $Y_1^f(t) = \mathbbm{1}(T\geqslant t)$. Then $M_1^f(t;\theta^1) = N_1^f(t)  - \int_0^tY_1^f(u) d\Lambda_1^1(u)$. We first compute $E\{N_1^f(t) \mid  \Tearly \geqslant r, A, Z\}$ and $E\{Y_1^f(t) \mid  \Tearly \geqslant r, A, Z\}$ when $t\geqslant r$ as follows: 
\begin{align*}
E\{N_1^f(t) \mid  \Tearly \geqslant r, A, Z\} &= P(T \leqslant t, T_1\leqslant T_2\mid \min(T, t)\geqslant r, A, Z)\\
&=  P(T \leqslant t, T_1\leqslant T_2\mid  T\geqslant r, A, Z)\\
&=  \frac{P(T \leqslant t, T_1\leqslant T_2, T \geqslant r\mid  A, Z)}{S(r| A,Z)}\\
&= \frac{P(r\leqslant T_1 \leqslant t, T_1\leqslant T_2\mid  A, Z)}{S(r| A,Z)}\\
&= \frac{\int_r^tS(u| A,Z)d\Lambda_1(u| A,Z) }{S(r| A,Z)}\\
E\{Y_1^f(t)\mid \Tearly \geqslant r, A, Z\} &= P(T\geqslant t\mid \min(T,t) \geqslant r, A,Z)\\
&= P(T\geqslant t\mid T \geqslant r, A,Z)\\
&= \frac{P(T\geqslant t\mid  A,Z)}{P(T\geqslant r\mid  A,Z)}\\
&= \frac{S( t|  A,Z)}{S(r|  A,Z)}
\end{align*}
Thus, 
\begin{equation*}
E\{M_1^f(t;\theta^1) \mid  \Tearly \geqslant r, A, Z\} = \frac{\int_r^tS(u| A,Z)d\Lambda_1(u| A,Z)-\int_r^tS(u| A,Z)d\Lambda_1^1(u)}{S(r| A,Z)}
\end{equation*}
and 
\begin{equation}
\label{eq:lambda1_h}
E\{U^f(t;\theta^1) \mid  \Tearly \geqslant r, A, Z\} = \frac{S(t| A,Z)\{d\Lambda_1(t| A,Z)-d\Lambda_1^1(t)\}}{S(r| A,Z)}. 
\end{equation}
Denote the numerator in \eqref{eq:lambda1_h} as $h_1(t, A, Z)$.
Then the second term in \eqref{eq:c-aipcw} becomes 
\begin{equation}\label{eq:lambda1_term2}
\left\{1-\frac{A}{\pi(Z)} \right\}h_1(t, 1, Z). 
\end{equation}
For the third term, we have 
\begin{align*}
\int_0^{\Xearly} \frac{E\{U^f(t;\theta^1) \mid  \Tearly \geqslant r, A=1, Z\}}{S_c(r |  A=1, Z)} d \causalMc(r; S_c) &= h_1(t, 1, Z)  \int_0^{\Xearly}\frac{d\causalMc(r;S_c)}{S(r| A=1,Z)S_c(r| A=1, Z)}.
\end{align*}
Let $ \tilde{X}_1 = \min(X_1, t)$.  We have  
\begin{align*}
\int_0^{\Xearly}\frac{d\causalMc(r;S_c)}{S(r| A,Z)S_c(r| A, Z)}  & = \frac{(1-\delta_2) \mathbbm{1}(X_2\leqslant \tilde{X}_1) }{ S(X_2 |  A, Z) S_c(X_2|  A, Z) } +\int_0^{\tilde{X}_1} \frac{d S_c(r |  A, Z) }{S(r |  A, Z) S_c^{2}(r|  A, Z) } \\
&\stackrel{\text {denote as}}{=}  m_1(\delta_2, X_1, X_2, A, Z, t).
\end{align*}
Thus the third term in \eqref{eq:c-aipcw} becomes 
\begin{equation}\label{eq:lambda1_term3}
\frac{A}{\pi(Z)}h_1(t, 1, Z) m_1(\delta_2, X_1, X_2, 1, Z, t)
\end{equation}
Combining the three terms in \eqref{eq:lambda1_term1}, \eqref{eq:lambda1_term2} and \eqref{eq:lambda1_term3} gives the following estimating function:
\begin{align}
U^1(t;\theta^1;S_c, \pi, F)&=\frac{A\{\mathbbm{1}(X_1=t)\delta_1-\mathbbm{1}(X_1 \geqslant t)d\Lambda_1^1(t)\}}{\pi(Z)S_c(t|A=1,Z)} \nonumber \\ 
&+\left\{1-\frac{A}{\pi(Z)}+\frac{A}{\pi(Z)} m_1(\delta_2, X_1, X_2, 1, Z, t)\right\}h_1(t, 1, Z). 
\end{align}
Solving $\sum_{i=1}^n U_i^1(t;\theta^1 ;\hat{S_c}, \hat{\pi},\hat{F})=0$ gives the expression of $\hat \Lambda_1^1(t)$ as in \eqref{eq:Lambda1}.

\bigskip
To estimate $\Lambda_3^1(t)$, we consider $\theta^1= \Lambda_3^1(t)$, $U^f(t; \theta^1) = dM_3^f(t;\theta^1)$
where $ M_3^f(t;\theta^1) \equiv  M_3^f(t;\theta^1;T_1, T_2) =\mathbbm{1}( T_1\leqslant T_2\leqslant t)-\int_0^t\mathbbm{1}(T_2\geqslant u > T_1) d\Lambda_3^1(u)$,
and $\Tearly = \min(T_2, t).$ Now we plug these into formula \eqref{eq:c-aipcw} which  consists of three terms. 
For the first term, we have
\begin{equation*}
U^{f}(t;\theta^1) = dM_3^f(t;\theta^1)=\mathbbm{1}(T_1\leqslant T_2=t)-\mathbbm{1}(T_2\geqslant t > T_1) d\Lambda_3^1(t)
\end{equation*} 
and 
\begin{equation}\label{eq:lambda3_term1}
\frac{A}{\pi(Z)}\frac{ \mathbbm{1}(\Tearly <C)} {S_c(\Xearly |  A=1, Z)} U^f(t ; \theta^1) = \frac{A\{\mathbbm{1}(X_2=t)\delta_1\delta_2-\mathbbm{1}(X_1 <t\leqslant X_2)\delta_1 d\Lambda_3^1(t)\}}{\pi(Z)S_c(t| A=1,Z)}.
\end{equation}
For the second and the third term, as mentioned in the previous derivation for $\hat \Lambda_1^1(t)$, we need to compute $E\{U^f(t;\theta^1) \mid  \Tearly \geqslant r, A=1, Z\}$. Let $N_3^f(t) = \mathbbm{1}(T_1\leqslant T_2\leqslant t)$ and $Y_3^f(t) = \mathbbm{1}(T_2\geqslant t> T_1)$. Then $M_3^f(t;\theta^1) = N_3^f(t)  - \int_0^tY_3^f(u) d\Lambda_3^1(u)$. We first compute $E\{N_3^f(t) \mid  \Tearly \geqslant r, A, Z\}$ and $E\{Y_3^f(t) \mid  \Tearly \geqslant r, A, Z\}$ when $t\geqslant r$ as follows: 
\begin{align*}
E\{N_3^f(t) \mid  \Tearly \geqslant r, A, Z\} &= P( T_1\leqslant T_2 \leqslant t \mid \min(T_2, t)\geqslant r, A, Z)\\
&=  P(T_1\leqslant T_2\leqslant t\mid  T_2\geqslant r, A, Z)\\
&=  \frac{P(r\leqslant T_2 \leqslant t, T_1\leqslant T_2\mid  A, Z)}{S_2(r|A,Z)}\\
E\{Y_3^f(t)\mid \Tearly \geqslant r, A, Z\} &= P(T_1<t \leqslant  T_2 \mid \min(T_2,t) \geqslant r, A,Z)\\
&= P(T_1<t \leqslant  T_2\mid T_2 \geqslant r, A,Z)\\
&= \frac{P(T_1<t \leqslant  T_2 \mid  A,Z)}{S_2(r|  A,Z)}\\
\end{align*}
Thus, 
\begin{equation*}
E\{M_3^f(t;\theta^1) | \Tearly \geqslant r, A=1, Z\} = \frac{P(r \leqslant T_2 \leqslant t, T_1\leqslant T_2 |  A, Z)-\int_0^t P(T_1<u \leqslant  T_2 | A=1,Z)d\Lambda_3^1(u)}{S_2(r| A=1,Z)}.
\end{equation*}
Since 
$$
P(T_1<t \leqslant  T_2\mid  A,Z) = S_2(t| A, Z) - S(t| A, Z),
$$
and 
\begin{align*}
P(T_1\leqslant T_2\leqslant t\mid A, Z) &= 1- P(T_2\leqslant  t, T_1=\infty \mid  A,Z) -S_2(t| A, Z)\\
&= 1- \int_0^t S(u| A,Z)d\Lambda_2(u | A,Z) -S_2(t| A, Z),
\end{align*}
we have 
\begin{equation}\label{eq:lambda3_h}
E\{U^f(t;\theta^1) |  \Tearly \geqslant r, A, Z\} = \frac{-S(t| A,Z)d\Lambda_2(t | A,Z) -dS_2(t| A, Z)-\{S_2(t| A,Z)-S(t| A,Z)\}d\Lambda_3^1(t)}{S_2(r| A,Z)}. 
\end{equation}
Denote the numerator in \eqref{eq:lambda3_h} as $h_3(t, A, Z)$.
Then the second term in \eqref{eq:c-aipcw} becomes 
\begin{equation}\label{eq:lambda3_term2}
\left\{1-\frac{A}{\pi(Z)} \right\}h_3(t, 1, Z). 
\end{equation}
For the third term, we have 
\begin{align*}
 \int_0^{\Xearly} \frac{E\{U^f(t;\theta^1) \mid  \Tearly \geqslant r, A=1, Z\}}{S_c(r |  A=1, Z)} d \causalMc(r; S_c) = h_3(t, 1, Z)\int_0^{\Xearly}\frac{d\causalMc(r;S_c)}{S_2(r| A=1,Z)S_c(r| A=1, Z)}.
\end{align*}
Let $ \tilde{X}_2 = \min(X_2, t)$.  We have  
\begin{align*}
\int_0^{\Xearly}\frac{d\causalMc(r;S_c)}{S_2(r| A,Z)S_c(r| A, Z)}  & = \frac{(1-\delta_2) \mathbbm{1}(X_2\leqslant \tilde{X}_2) }{ S_2(X_2 |  A, Z) S_c(X_2|  A, Z) } +\int_0^{\tilde{X}_2} \frac{d S_c(r |  A, Z) }{S_2(r | A, Z) S_c^{2}(r|  A, Z) } \\
&\stackrel{\text {denote as}}{=}  m_3(\delta_2, X_2, A, Z, t).
\end{align*}
Thus the third term in \eqref{eq:c-aipcw} becomes 
\begin{equation}\label{eq:lambda3_term3}
\frac{A}{\pi(Z)} h_3(t, 1, Z) m_3(\delta_2, X_2, 1, Z, t)
\end{equation}
Combining the three terms in \eqref{eq:lambda3_term1}, \eqref{eq:lambda3_term2} and \eqref{eq:lambda3_term3} gives the following estimating function:
\begin{align}
U^1(t;\theta^1;S_c, \pi, F)&= \frac{A\{\mathbbm{1}(X_2=t)\delta_1\delta_2-\mathbbm{1}(X_1 <t\leqslant X_2)\delta_1 d\Lambda_3^1(t)\}}{\pi(Z)S_c(t| A=1,Z)} \nonumber \\ 
&+\left\{1-\frac{A}{\pi(Z)}+   \frac{A}{\pi(Z)} m_3(\delta_2, X_2, 1, Z, t)\right\} h_3(t, 1, Z). 
\end{align}
Solving $\sum_{i=1}^n U_i^1(t;\theta^1 ;\hat{S_c}, \hat{\pi},\hat{F})=0$ gives the expression of $\hat \Lambda_3^1(t)$ as in \eqref{eq:Lambda3}.

\subsection{IPW and naive Nelson-Aalen estimators}
For comparison,  we write out the corresponding IPW estimators that are not doubly robust. We also consider the naive Nelson-Aalen estimators (denoted as NA) which ignore confounding and covariate-dependent censoring. The expressions of these estimators are as follows.
For $a=0,1:$
\begin{align*}
\hat{\Lambda}_{1, \text{IPW}}^a(t) &= \int_0^t\frac{
\sum_{i=1}^n Q_{1i}^a(u)dN_{1i}(u) }{
\sum_{i=1}^n Q_{1i}^a(u)Y_{1i}(u)},\quad 
\hat{\Lambda}_{2, \text{IPW}}^a(t) =\int_0^t \frac{\sum_{i=1}^n Q_{1i}^a(u)dN_{2i}(u) }{
\sum_{i=1}^n Q_{1i}^a(u)Y_{1i}(u)},\\
\hat{\Lambda}_{3, \text{IPW}}^a(t)&=\int_{0}^{t}\frac{
\sum_{i=1}^n Q_{1i}^a(u)dN_{3i}(u)}{
\sum_{i=1}^n Q_{1i}^a(u)Y_{3i}(u)},\\
\hat{\Lambda}_{1, \text{NA}}^a(t) &= \int_0^t\frac{
\sum_{i=1}^n A_i^a(1-A_i)^{(1-a)}dN_{1i}(u) }{
\sum_{i=1}^n A_i^a(1-A_i)^{(1-a)}Y_{1i}(u)},\quad 
\hat{\Lambda}_{2, \text{NA}}^a(t) =\int_0^t \frac{\sum_{i=1}^n A_i^a(1-A_i)^{(1-a)}dN_{2i}(u) }{
\sum_{i=1}^n A_i^a(1-A_i)^{(1-a)}Y_{1i}(u)},\\
\hat{\Lambda}_{3, \text{NA}}^a(t)&=\int_{0}^{t}\frac{
\sum_{i=1}^n A_i^a(1-A_i)^{(1-a)}dN_{3i}(u)}{
\sum_{i=1}^n A_i^a(1-A_i)^{(1-a)}Y_{3i}(u)}.\\
\end{align*}

\subsection{Simulation details}\label{append:sim_data}
The semi-competing risks data were simulated using an approach adapted from \cite{zhang2024marginal} in the following steps, with an illustration in Figure \ref{fig:sim_dag}:
\begin{enumerate}
    \item Generate $U_1 \sim U(0,1)$ and $U_2 \sim U(0,1)$;

    \item Generate confounder $Z=(Z_1, Z_2)^{\top}$, with $Z_j=U_j-0.5,$ $j=1,2$

    \item Let $f^{-1}$ denote the inverse function of a function $f$. Generate $A \sim \operatorname{Bernoulli}(p_A)$, where $p_A=\operatorname{logit}^{-1}(\alpha_0+\alpha_1 Z_1+\alpha_2 Z_2)$, with $(\alpha_0,  \alpha_1, \alpha_2) = (0,1.5,-1.5)$;
\item Let $\Lambda_{0}^0(t) = 0.2t$, $\Lambda_{0}^1(t) = 0.3t$. Generate $T_1^a$, $T_2^a$ as follows: \\
with probability 1/2,
\[
T_1^a=\infty,\quad
T_2^a=\Lambda_{0}^{a, -1}\left(-\frac{\log (U_1)}{2}\right)
\]
with probability 1/2,
\[
T_1^a=\Lambda_{0}^{a, -1}(-\frac{\log(U_1)}{2}),\quad T_2^a=\Lambda_{0}^{a, -1}(-\log (U_2)+\Lambda_{0}^a(t_1))
\]
\item Generate censoring time $C^a\mid  Z \sim \text{Exp}(\exp(\gamma_0+\gamma_1Z_1+\gamma_2Z_2 +\gamma_3a))$, where $\text{Exp}$ denotes the exponential distribution and $(\gamma_0, \gamma_1, \gamma_2, \gamma_3) = (-2.5, 1, 1, 0.5)$. 
\item Generate $T_1 = AT_1^1+(1-A)T_1^0, T_2 = AT_2^1+(1-A)T_2^0, C  = AC^1+(1-A)C^0.$
\item Generate $(X_1, X_2,\delta_1,\delta_2)$ from $(T_1, T_2, C)$ accordingly. 
\end{enumerate}
\begin{figure}[!htbp]
\centering
\includegraphics[width=0.6\linewidth]{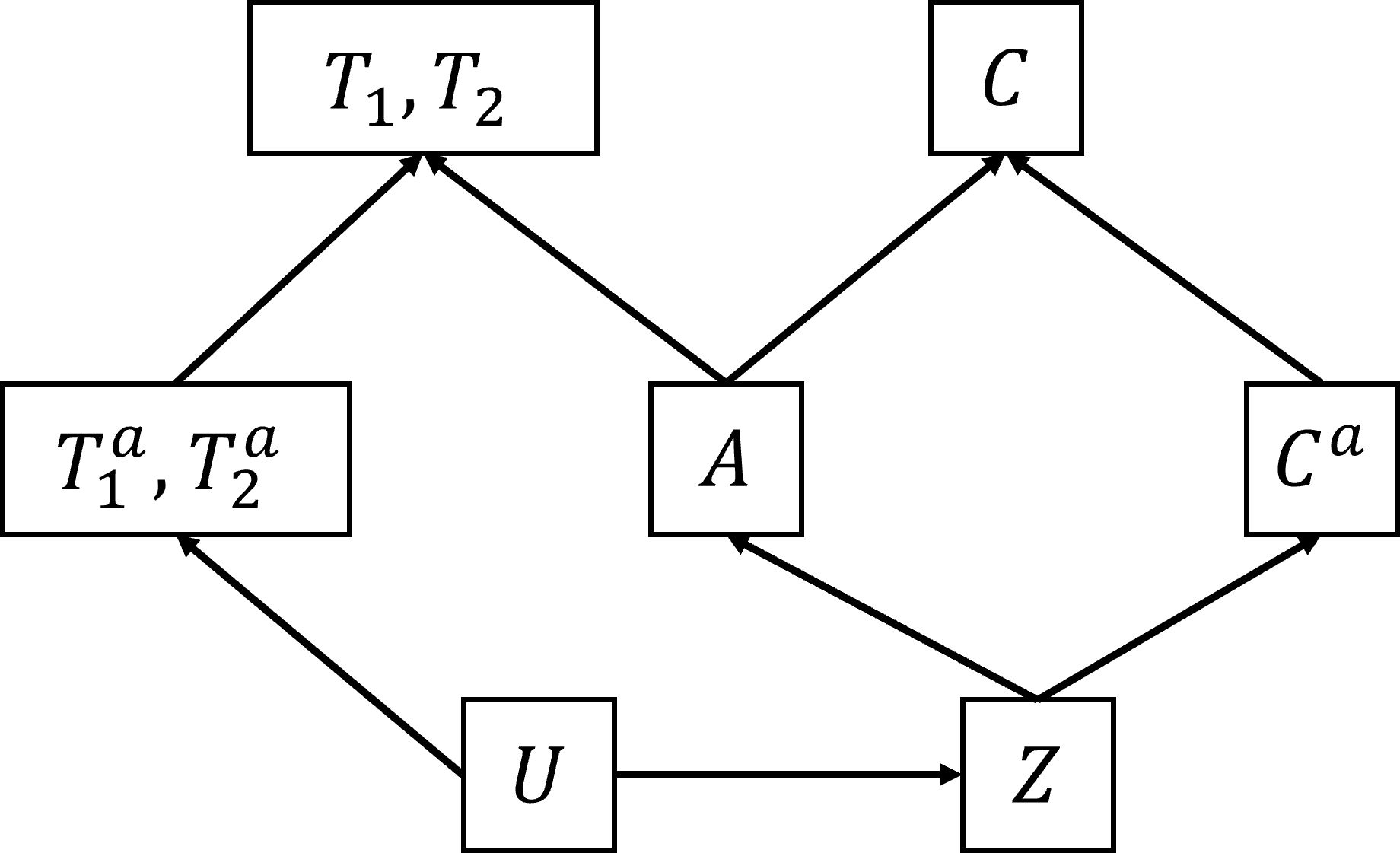}
\caption{Diagram of data generation mechanism in simulations} 
\label{fig:sim_dag}
\end{figure}
\clearpage

\subsection{Additional simulation results}\label{append:sim_results}
In our simulation studies, we considered as estimands $\operatorname{CIF}_1^a(3)$, $\operatorname{CIF}_2^a(3)$ and  $\operatorname{CIF}_3^a(3 | 0.5)$ for $a=0,1$, resulting in six estimands in total. The results for $\operatorname{CIF}_j^0, j=1,2,3$ are presented in the main text. The results for $\operatorname{CIF}_j^1, j=1,2,3$ are presented in Table \ref{append:table_sim}. 
\begin{table}
\caption{Simulation results. Models in red are misspecified: for the event time model, it does not follow a Cox specification based on the data generation mechanism; for the censoring model, Cox1 is the correct model $C\sim A+Z_1+Z_2$ and Cox2 is the misspecified model $C\sim A+Z_1^2+Z_2^2$; for the PS model, logit1 is the correct model $A\sim Z_1+Z_2$ and logit2 is the misspecified model $A\sim Z_1^2+Z_2^2$.}
\label{append:table_sim}
\centering
\begin{tabular}[t]{lllrlll}
\toprule
Estimand:truth & Method & Event Time/Censor-PS & Bias & SD & SE & CP(\%)\\
\midrule
$\text{CIF}_1^1(3)$: 0.431 & Causal- & RSF/(Cox1-gbm) & -0.002 & 0.033 & 0.032 & 95.4\\
 & AIPCW\textsuperscript{sc} & RSF/(Cox1-logit1) & -0.004 & 0.033 & 0.031 & 94.0\\
 &  & RSF/(RSF-gbm) & -0.001 & 0.032 & 0.032 & 95.9\\
 &  & \textcolor{red}{Cox}/(Cox1-logit1) & -0.004 & 0.034 & 0.033 & 93.9\\
 &  & \textcolor{red}{Cox}/(RSF-gbm) & 0.001 & 0.033 & 0.033 & 95.4\\
 &  & \textcolor{red}{Cox}/(RSF-logit1) & -0.003 & 0.034 & 0.032 & 94.3\\
 &  & \textcolor{red}{Cox}/(\textcolor{red}{Cox2}-\textcolor{red}{logit2}) & 0.017 & 0.033 & 0.032 & 91.7\\
  \cline{2-7}
 & IPW & -/(Cox1-logit1) & -0.005 & 0.035 & 0.033 & 93.8\\
 &  & -/(RSF-gbm) & 0.002 & 0.033 & 0.033 & 95.8\\
 &  & -/(RSF-logit1) & -0.004 & 0.034 & 0.033 & 94.6\\
 &  & -/(\textcolor{red}{Cox2}-\textcolor{red}{logit2}) & 0.021 & 0.034 & 0.033 & 90.6\\
  \cline{2-7}
 & Naive & - & 0.019 & 0.034 & 0.033 & 91.1\\
 & Full & - & -0.001 & 0.022 & 0.022 & 94.6\\
\addlinespace
$\text{CIF}_2^1(3)$: 0.388 & Causal- & RSF/(Cox1-gbm) & 0.000 & 0.034 & 0.034 & 95.5\\
 & AIPCW\textsuperscript{sc} & RSF/(Cox1-logit1) & 0.002 & 0.034 & 0.032 & 94.0\\
 &  & RSF/(RSF-gbm) & 0.001 & 0.034 & 0.034 & 95.1\\
 &  & \textcolor{red}{Cox}/(Cox1-logit1) & 0.000 & 0.035 & 0.034 & 94.9\\
 &  & \textcolor{red}{Cox}/(RSF-gbm) & -0.001 & 0.034 & 0.034 & 95.3\\
 &  & \textcolor{red}{Cox}/(RSF-logit1) & 0.002 & 0.035 & 0.034 & 94.6\\
 &  & \textcolor{red}{Cox}/(\textcolor{red}{Cox2}-\textcolor{red}{logit2}) & -0.014 & 0.034 & 0.033 & 92.3\\
  \cline{2-7}
 & IPW & -/(Cox1-logit1) & -0.001 & 0.035 & 0.034 & 95.4\\
 &  & -/(RSF-gbm) & -0.005 & 0.034 & 0.034 & 94.6\\
 &  & -/(RSF-logit1) & -0.001 & 0.035 & 0.033 & 94.6\\
 &  & -/(\textcolor{red}{Cox2}-\textcolor{red}{logit2}) & -0.030 & 0.032 & 0.031 & 82.1\\
  \cline{2-7}
 & Naive & - & -0.033 & 0.032 & 0.032 & 82.3\\
 & Full & - & 0.001 & 0.022 & 0.022 & 95.4\\
\addlinespace
$\text{CIF}_3^1(3|0.5)$: 0.527 & Causal- & RSF/(Cox1-gbm) & 0.012 & 0.054 & 0.056 & 96.2\\
 & AIPCW\textsuperscript{sc} & RSF/(Cox1-logit1) & 0.014 & 0.055 & 0.052 & 94.3\\
 &  & RSF/(RSF-gbm) & 0.011 & 0.054 & 0.057 & 95.5\\
 &  & \textcolor{red}{Cox}/(Cox1-logit1) & 0.002 & 0.054 & 0.052 & 94.6\\
 &  & \textcolor{red}{Cox}/(RSF-gbm) & -0.007 & 0.052 & 0.054 & 95.9\\
 &  & \textcolor{red}{Cox}/(RSF-logit1) & 0.001 & 0.053 & 0.052 & 95.3\\
 &  & \textcolor{red}{Cox}/(\textcolor{red}{Cox2}-\textcolor{red}{logit2}) & -0.041 & 0.049 & 0.049 & 85.5\\
  \cline{2-7}
 & IPW & -/(Cox1-logit1) & -0.005 & 0.060 & 0.059 & 94.5\\
 &  & -/(RSF-gbm) & -0.017 & 0.057 & 0.060 & 95.2\\
 &  & -/(RSF-logit1) & -0.002 & 0.058 & 0.056 & 94.3\\
 &  & -/(\textcolor{red}{Cox2}-\textcolor{red}{logit2}) & -0.121 & 0.058 & 0.058 & 42.7\\
 \cline{2-7}
 & Naive & - & -0.121 & 0.058 & 0.059 & 44.7\\
 & Full & - & 0.000 & 0.065 & 0.064 & 95.9\\
\bottomrule
\end{tabular}
\end{table}

\clearpage

\subsection{Application to HAAS Study}\label{append:haas}
We presented the plots of estimated treatment-specific risks and risk differences in the main text. The corresponding numeric results are presented in Table \ref{table:haas}.
\begin{table}[!h]
\centering
\caption{Estimated treatment-specific risks and risk differences with their 95\% confidence intervals for moderate impairment (MI), death without MI, and death after MI by 8 years\label{table:haas}}
\centering
\begin{tabular}[t]{lllll}
\toprule
 & Time & Non-heavy Drinking & Heavy Drinking & Risk Difference\\
\midrule
 & 5 & 0.133 (0.115, 0.152) & 0.151 (0.117, 0.185) & 0.017 (-0.020, 0.055)\\

 & 10 & 0.297 (0.274, 0.320) & 0.367 (0.319, 0.415) & 0.069 (0.015, 0.123)*\\

 & 15 & 0.414 (0.390, 0.439) & 0.480 (0.428, 0.532) & 0.066 (0.008, 0.123)*\\

\multirow[t]{-4}{*}{\raggedright\arraybackslash MI} & 20 & 0.546 (0.511, 0.580) & 0.568 (0.520, 0.616) & 0.022 (-0.039, 0.084)\\

 & 5 & 0.047 (0.037, 0.058) & 0.075 (0.048, 0.103) & 0.028 (-0.001, 0.057)\\

 & 10 & 0.198 (0.178, 0.219) & 0.235 (0.195, 0.274) & 0.037 (-0.009, 0.082)\\

 & 15 & 0.319 (0.295, 0.343) & 0.328 (0.282, 0.373) & 0.009 (-0.043, 0.060)\\

\multirow[t]{-4}{*}{\raggedright\arraybackslash Death without MI} & 20 & 0.398 (0.372, 0.424) & 0.380 (0.333, 0.427) & -0.018 (-0.073, 0.036)\\

 & 10 & 0.248 (0.204, 0.292) & 0.281 (0.196, 0.366) & 0.033 (-0.064, 0.129)\\

 & 15 & 0.703 (0.654, 0.752) & 0.749 (0.681, 0.818) & 0.046 (-0.040, 0.132)\\

\multirow[t]{-3}{*}{\raggedright\arraybackslash Death after MI} & 20 & 0.901 (0.848, 0.955) & 0.953 (0.904, 1.000) & 0.052 (-0.019, 0.123)\\
\bottomrule
\end{tabular}
\parbox{0.9\linewidth}{\small
$^{*}$ statistically significant  at $\alpha = 0.05$ level two-sided}
\end{table}

\subsection{Conditional Markov does not imply marginal Markov}\label{append:Markov}

We show that conditional Markov assumption does not imply marginal Markov assumption. 
Let 
$\lambda_3(t \mid t_1) = \lim _{\Delta \rightarrow 0} P(T_2 \in[t, t+\Delta) \mid  T_1=t_1,T_2 \geqslant t)/\Delta$ denote the  transition rate from the non-terminal event to the terminal event. Let 
$\lambda_3(t \mid t_1, Z) = \lim _{\Delta \rightarrow 0} P(T_2 \in[t, t+\Delta) \mid  T_1=t_1,T_2 \geqslant t, Z)/\Delta$ denote the corresponding transition rate conditional on $Z$. By conditional Markov assumption, we have
$\lambda_3(t| t_1, Z)=\lambda_3(t|Z)\mathbbm{1}(t>t_1) $. 
Now we compute $\lambda_3(t \mid t_1)$ based on $\lambda_3(t\mid t_1, Z)$. To facilitate the computation, let us suppose $\lambda_3(t,Z) = \lambda_0(t)\alpha(Z)$. 
Then when $ t>t_1$:
\begin{align*}
\lambda_3(t \mid t_1) & =\lim _{\Delta \to 0} \frac{1}{\Delta} \frac{P\{T_2 \in[t, t+\Delta) \mid T_1=t_1\}}{P(T_2 \geqslant t \mid T_1=t_1)}\\
& =\lim _{\Delta \to 0} \frac{1}{\Delta} \frac{E[P\{T_2 \in[t, t+\Delta) \mid T_1=t_1,Z\}]}{E[P(T_2 \geqslant t \mid T_1=t_1,Z)]} \quad \text{[by tower rule]}\\
& =\lim _{\Delta \to 0} \frac{1}{\Delta} \frac{E\left[\exp \left\{-\int_{t_1}^t \lambda_3(u \mid t_1, Z) d u\right\}-\exp \left\{-\int_{t_1}^{t+\Delta} \lambda_3(u \mid t_1, Z) d u\right\}\right]}{E\left[\exp \left\{-\int_{t_1}^t \lambda_3(u \mid t_1, Z) d u\right\}\right]} \\
& =\lim _{\Delta \to 0} \frac{1}{\Delta} \frac{E\left[\exp \left\{-\alpha(Z)\int_{t_1}^t \lambda_0(u) d u\right\}-\exp \left\{-\alpha(Z)\int_{t_1}^{t+\Delta} \lambda_0(u) d u\right\}\right]}{E\left[\exp \left\{-\alpha(Z)\int_{t_1}^t \lambda_0(u) d u\right\}\right]}
\end{align*}
Let  $g(a, b)=\int_a^b \lambda_0(u) d u$, $m(Z)=\exp \{-\alpha(Z)\}$. Then 
\begin{equation}\label{eq:markov}
\lambda_3(t \mid t_1) =\lim _{\Delta \to 0} \frac{1}{\Delta} \left[1-\frac{E\left\{m(Z)^{g(t_1,t)+g(t,t+\Delta)} \right\}}{E\left\{m(Z)^{g(t_1,t)} \right\}}\right]
\end{equation}
Note that \eqref{eq:markov} depends on $t_1$ in general.
For example, let $\alpha(Z)=Z$,  $\lambda_0(t)=1$ and $Z$ follow a uniform distribution from 0 to 1. 
Then clearly $\lambda_3(t \mid t_1, Z)=Z \mathbbm{1}(t>t_1)$ satisfies Markov assumption conditional on $Z$.
However for $t>t_1$,
\begin{align}
\lambda_3(t \mid t_1) &=\lim _{\Delta \to 0} \frac{1}{\Delta} \left[1-\frac{E\left\{\exp(-Z)^{(t+\Delta-t_1)} \right\}}{E\left\{\exp(-Z)^{(t-t_1)} \right\}}\right]  \nonumber\\
& =\lim_{\Delta \to 0}-\frac{t-t_1}{1-\exp(t_1-t)}\left\{\frac{1-\exp(-t-\Delta+t_1)}{t+\Delta-t_1}\right\}' \quad [\text{L'Hôpital's Rule]} \nonumber\\
& =\frac{1}{t-t_1}-\frac{1}{e^{t-t_1}-1} . \label{eq:markov_example}
\end{align}
Note that \eqref{eq:markov_example} clearly depends on $t_1$, which violates the marginal Markov assumption. Thus, conditional Markov does not imply marginal Markov.
\end{document}